%% file: main.tex
\documentclass[a4paper,UKenglish,cleveref, autoref, thm-restate]{lipics-v2021}

\pdfoutput=1
\nolinenumbers

\usepackage{comment}

\title{
Sorting Finite Automata via Partition Refinement 
} 

\author{Ruben {Becker}}{Ca' Foscari University of Venice, Italy}{rubensimon.becker@unive.it}{https://orcid.org/0000-0002-3495-3753}{}
\author{Manuel {C\'aceres}}{University of Helsinki, Finland \and \url{https://me.ariel.computer/}}{manuel.caceresreyes@helsinki.fi}{https://orcid.org/0000-0003-0235-6951}{Funded by the Academy of Finland (grants No. 352821, 328877).} 
\author{Davide {Cenzato}}{Ca' Foscari University of Venice, Italy}{davide.cenzato@unive.it}{https://orcid.org/0000-0002-0098-3620}{}
\author{Sung-Hwan {Kim}}{Ca' Foscari University of Venice, Italy}{sunghwan.kim@unive.it}{https://orcid.org/0000-0002-1117-5020}{}
\author{Bojana {Kodric}}{Ca' Foscari University of Venice, Italy}{bojana.kodric@unive.it}{https://orcid.org/0000-0001-7242-0096}{}
\author{Francisco {Olivares}}{CeBiB --- Centre for Biotechnology and Bioengineering \and
Department of Computer Science, University of Chile, Chile}{folivares@uchile.cl}{https://orcid.org/0000-0001-7881-9794}{Funded by Ph.D Scholarship 21210579, ANID, Chile.}
\author{Nicola {Prezza}}{Ca' Foscari University of Venice, Italy}{nicola.prezza@unive.it}{https://orcid.org/0000-0003-3553-4953}{}

\funding{\textit{Ruben Becker, Davide Cenzato, Sung-Hwan Kim, Bojana Kodric, Nicola Prezza}: Funded by the European Union (ERC, REGINDEX, 101039208). Views and opinions expressed are however those of the author(s) only and do not necessarily reflect those of the European Union or the European Research Council. Neither the European Union nor the granting authority can be held responsible for them.}

\authorrunning{Ruben Becker et al.}

\Copyright{Ruben Becker, Manuel C\'aceres, Davide Cenzato, Sung-Hwan Kim, Bojana Kodric, Francisco Olivares, and Nicola Prezza} 



\category{} 

\relatedversion{This is the full version of the paper \cite{ESAsubmission} that appeared in the Proceedings of the 31st Annual European Symposium on Algorithms (ESA 2023).}

\acknowledgements{}

\usepackage{environ}
\input{macros}

\EventEditors{Inge Li G{\o}rtz, Martin Farach-Colton, Simon J. Puglisi, and Grzegorz Herman}
\EventNoEds{4}
\EventLongTitle{31st Annual European Symposium on Algorithms (ESA 2023)}
\EventShortTitle{ESA 2023}
\EventAcronym{ESA}
\EventYear{2023}
\EventDate{September 4--6, 2023}
\EventLocation{Amsterdam, the Netherlands}
\EventLogo{}
\SeriesVolume{274}
\ArticleNo{15}

\begin{document}

\maketitle

\begin{abstract}
    Wheeler nondeterministic finite automata (WNFAs) were introduced in (Gagie et al., TCS 2017) as a powerful generalization of prefix sorting from strings to labeled graphs. WNFAs admit optimal solutions to classic hard problems on labeled graphs and languages such as compression and regular expression matching. 
    The problem of deciding whether a given NFA is Wheeler is known to be NP-complete (Gibney and Thankachan, ESA 2019). Recently, however, Alanko et al.\ (Information and Computation 2021) showed how to side-step this complexity by switching to \emph{preorders}: letting $Q$ be the set of  states and $\delta$ the set of transitions,
    they provided a
    $O(|\delta|\cdot|Q|^2)$-time algorithm 
    computing a totally-ordered \emph{partition} (i.e.\ equivalence relation) of the WNFA's states such that 
    (1) equivalent states recognize the same regular language, and
    (2) the order of (the classes of) non-equivalent states 
    is consistent with any Wheeler order, when one exists. 
    As a result, the output is a preorder of the states as useful for pattern matching as standard Wheeler orders.
    
    Further extensions of this line of work (Cotumaccio et al., SODA 2021 and DCC 2022) generalized these concepts to arbitrary NFAs by introducing \emph{co-lex partial preorders}: in general, any NFA admits a partial preorder of its states reflecting the co-lexicographic order of their accepted strings; the smaller the width of such preorder is, the faster regular expression matching queries can be performed. To date, the fastest algorithm for computing the smallest-width partial preorder on NFAs runs in 
    $O(|\delta|^2 + |Q|^{5/2})$ time (Cotumaccio, DCC 2022), while on DFAs the same task can be accomplished in
    $O(\min(|Q|^2\log|Q|, |\delta|\cdot |Q|))$ time (Kim et al., CPM 2023). 

    In this paper, we provide much more efficient solutions to the 
    co-lex order computation problem. Our results are achieved by
    extending a classic algorithm for the relational coarsest partition refinement problem of Paige and Tarjan
    to work with \emph{ordered} partitions.
    More specifically, we provide a $O(|\delta| \log |Q|)$-time algorithm computing a co-lex total preorder when the input is a Wheeler NFA, and an algorithm with the same time complexity computing the smallest-width co-lex partial order of any DFA. 
    In addition, we present implementations of our algorithms and show that they are very efficient also in practice.
\end{abstract}

\section{Introduction}

The classical \emph{pattern matching} problem between two strings $S$ (the text) and $P$ (the pattern) over alphabet $\Sigma$ asks to find the substrings of $S$ matching $P$. Although many algorithms solving the \emph{on-line} version of the problem exist, in many scenarios it is possible to pre-process $S$ \emph{off-line} into an \emph{index} to speed up subsequent pattern matching queries. As a matter of fact, a very  successful line of research dating back to the invention of suffix trees~\cite{weiner1973linear} and culminating with the discovery of compressed data structures~\cite{navarro:cup16} showed that it is possible to represent $S$ in compact space while speeding up matching queries.

The \emph{indexed pattern matching} problem can be generalized to collections of strings: in this case the pattern must be found as a substring in a string collection $\mathcal S$. 
A natural approach to solve this problem is to concatenate all strings in $\mathcal S$ into one string $S$ and re-use the well-optimized techniques for classical pattern matching. Even though
there are many successful examples of indexes following this strategy
\cite{gagie:latin14, gagie:soda18, makinen:jcb10}, this approach suffers from high space consumption during index construction (the input is very large), and does not scale to a more general (and interesting) scenario: the case where $\mathcal S$ contains an \emph{infinite} number of strings. 
A solution, addressing both issues, is to represent a (potentially infinite) collection of strings using a finite-state automaton with set of states $Q$ and transition function $\delta$.
In this new scenario, the goal of pattern matching is to locate walks in the automaton (seen as a labeled graph) spelling the query pattern $P$.
In bioinformatics, for example, genomic collections are represented using pangenome graphs: labeled graphs encoding nucleotide variations within the collection~\cite{computational2018computational,siren:ieee14}. Matching a DNA sequence over a pangenome graph allows one to discover the genetic variation in a population~\cite{siren2021pangenomics}.

Unfortunately, both the off-line and on-line pattern matching problems are hard to solve on labeled graphs: Equi et~al.~\cite{equi:sofsem21,equi2023complexity} showed that, conditioned on OVH~\cite{williams2005new}, it is not possible to design a polynomial-time algorithm to index a labeled graph such that pattern matching queries can be answered in $O(|P| + |\delta|^{\alpha}|P|^{\beta})$ time, for any constants $\alpha < 1$ or $\beta < 1$.

A successful paradigm to cope with this hardness is to study sub-classes of graphs where the problem is easier. Along these lines, Gagie et al.~\cite{gagie2017wheeler} introduced the class of \textit{Wheeler automata}: 
labeled graphs admitting a total order of their states (a \emph{Wheeler order}) which respects the underlying alphabet's order and it is propagated through pairs of equally-labeled edges. Wheeler orders generalize prefix sorting (the machinery standing at the core of the most successful string indexes) to labeled graphs and as a consequence, an index over a Wheeler automaton supports pattern matching queries in near-optimal $O(|P| \log |\Sigma|)$ time. 

 Wheeler automata, however, have two important limitations. First, 
the classes of 
Wheeler automata and 
Wheeler languages (accepted by such automata) are quite restricted. For example, Wheeler automata cannot contain proper cycles labeled with a unary string, and moreover, Wheeler languages are star-free~\cite{alanko:iac21}. 
Second, several natural problems related to Wheeler graphs are computationally hard. For example, the simple fact of deciding if an automaton is Wheeler is NP-complete, even when the automaton is acyclic~\cite{GibneyT19}.

Related to the first issue, the work~\cite{CotumaccioP21} extended state-ordering to arbitrary automata by using the concept of \emph{co-lex partial order}. By switching from total to partial orders, the authors showed that (i) every automaton can be (partially) sorted, and that (ii) the efficiency of pattern matching on the automaton depends on the \emph{width} (maximum size of an antichain) of such partial order. Wheeler automata are the particular case in which there exists a total co-lex order (i.e., of width one), thus enabling near-optimal time pattern matching queries.

A way to circumvent the latter limitation (hardness of computing a Wheeler order) was proposed by Alanko et al.~\cite{alanko:iac21} by switching to \emph{total preorders}: the authors showed a $O(|\delta|\cdot |Q|^2)$-time \emph{partition refinement} algorithm, which outputs a totally-ordered partition of $Q$ such that (1) states in the same part recognize the same language, and (2) the order of the partition is consistent with \emph{any} Wheeler order of $Q$. In other words, this ordered partition is a preorder of the states as useful for indexing as Wheeler orders. 
Recently, a similar solution was proposed by Chao et~al.~\cite{chao2022wgt} to speed up the computation of Wheeler orders in practice; after a first polynomial-time partition refinement step, their tool runs an exponential-time solver to assign a Wheeler ordering within classes of equivalent states. 

These strategies --- partial orders and total preorders --- were finally merged by Cotumaccio~\cite{Cotumaccio:dcc22}. As in \cite{CotumaccioP21}, the efficiency of pattern matching depends on the width of such a partial preorder. The author described a polynomial-time algorithm computing a \emph{colex partial preorder} of smallest width over an arbitrary nondeterministic finite-state automaton (NFA) in $O(|\delta|^2 + |Q|^{5/2})$ time. Later, this running time was improved by Kim et al. \cite{KimOP23} in the particular case of DFAs with two algorithms running in time $O(|Q|^2\log|Q|)$ and $O(|\delta|\cdot |Q|)$, respectively, and one algorithm running in near-optimal time 
 $O(|\delta|\log|Q|)$ on acyclic DFAs. 
Within the same running times, all the above-mentioned algorithms compute also a chain partition of minimum size $p$ of the partial preorder (by Dilworth's theorem \cite{dilworth1987decomposition}, $p$ is equal to the order's width); such a chain partition is needed to build the index described in \cite{Cotumaccio:dcc22}. The index can be built in linear time given as input such a chain partition, uses $(\lceil \log|\Sigma| \rceil + \lceil \log p \rceil + 2) \cdot (1 + o(1))$ bits per edge, and answers pattern matching queries on the regular language accepted by the automaton in $O(p^2\log(p \cdot |\Sigma|))$ 
time per pattern's character. 

\subparagraph*{Our contributions.} 
We consider two cases of colex orders: (1)~total preorders on NFAs, and (2)~partial orders of minimum width on DFAs. We provide algorithms running in time $O(|\delta|\log|Q|)$ for both cases, improving the state-of-the-art cubic and quadratic algorithms to near-optimal time.

Our solution to (1) is obtained by extending a classic algorithm for the relational coarsest partition refinement  problem of Paige and Tarjan~\cite{PaigeT87} to work with \emph{ordered} partitions. 
Our algorithm starts from the ordered partition corresponding to the states' incoming labels and, similarly to~\cite{PaigeT87}, iteratively refines this partition by enforcing \emph{forward-stability}: for any two parts $B$ (the ``splitter'') and $D$ (the ``split'') the image $\delta_a(B)$ of $B$ through the transition function (for any alphabet's character $a$) must either contain $D$ or be disjoint with $D$. If this condition is not satisfied, $D$ is split into three parts: states reached only from $B$, 
states reached only from $Q \setminus B$, and states reached from both. In addition to Paige and Tarjan's algorithm, we show how to sort these three parts consistently with the current partition's order. In fact, we show that our algorithm finds a total preorder on a class of automata being strictly larger than the Wheeler automata. We dub this class \emph{quasi-Wheeler}, thereby extending the class of automata that can be indexed to support queries in near-optimal time.
\begin{restatable}{theorem}{wheelerPreorder}
\label{theorem: wheeler preorder}
    For an NFA $\AAA = (Q, \delta, s)$, we can compute a total preorder $\preceq$ in $O(|\delta| \log |Q|)$ time such that $\preceq$ corresponds to a Wheeler preorder if and only if $\AAA$ is quasi-Wheeler.
\end{restatable}

Our solution to (2) relies on a further modification of the partition refinement algorithm. More specifically, we show that within the same time complexity, we can prune the automaton's transition function so that the resulting graph is a quasi-forest satisfying the following property: every state $u$ has only one incoming walk, whose label is co-lexicographically smallest (the \emph{infimum}) among the labels of all walks ending at $u$ in the original DFA. Symmetrically, we compute a pruned automaton encoding the co-lexicographically largest strings (the \emph{suprema}).  
We obtain our second result by plugging in the $O(|Q|\log |Q|)$-time algorithms of Kim et al.~\cite{KimOP23}, which sort the infimum and suprema strings (suffix-doubling)\footnote{More precisely, it is a special case of the suffix doubling algorithm \cite{KimOP23}, which runs in the claimed time.} and then compute a minimum chain partition of the smallest-width co-lex order (interval graph coloring). Our result assumes that the DFA is input-consistent (all in-going transitions to a state are labeled with the same letter). This assumption is a common simplification in automata theory and it is not a limitation of our approach as every DFA can be easily transformed into an equivalent input-consistent DFA.

\begin{restatable}{theorem}{wheelerPoset}\label{theorem: wheeler poset}
For an input-consistent DFA $\AAA=(Q,\delta,s)$, we can compute a minimum chain partition of the smallest-width co-lex order in $O(|\delta|\log|Q|)$ time.
\end{restatable}

We have implemented our partition refinement algorithms. We compared the implementation of our algorithm from \Cref{theorem: wheeler preorder} with a heuristic implementation from \texttt{WGT}~\cite{chao2022wgt} (their first polynomial time step) which, similarly to our algorithm, computes an ordered partition being consistent with any Wheeler order (although not the most refined --- unlike the output of our algorithm). On random Wheeler NFAs generated with a tool from the same toolbox (\texttt{WGT}), our implementation outperforms the heuristic from \texttt{WGT} by more than two orders of magnitude. We also experimentally show that our algorithm from \Cref{theorem: wheeler poset} can prune a pangenomic DFA with more than 50 million states in less than 6 minutes.

\section{Notation and Preliminaries}
For an integer $k\ge 1$, we let $[k]:=\{1, \ldots, k\}$. Given a set $U$, a partition $\PPP$ of $U$ is a set of pairwise disjoint non-empty sets $\{U_1, \ldots, U_k\}$ whose union is $U$. We call $U_1, \ldots, U_k$ the \emph{parts} of $\PPP$. If we, in addition, have a (total) order of the parts, we say that it is an \emph{ordered partition} and denote it with $\langle U_1, \ldots, U_k\rangle$. For two (ordered) partitions $\PPP$ and $\PPP'$ of $U$, we say that $\PPP'$ is a \emph{refinement} of $\PPP$ if every part of $\PPP'$ is contained in a part of $\PPP$, i.e., for every $U' \in \PPP'$ there is $U \in \PPP$ with $U'\subseteq U$. As a special case, a partition is a refinement of itself.

\subparagraph{Relations and Orders.}
A \emph{relation} $R$ over a set $U$ is a set of ordered pairs of elements from $U$, i.e., $R\subseteq U\times U$. We sometimes omit $U$ if it is clear from the context. For two elements $u,v$ from $U$, we usually write $uRv$ for $(u,v)\in R$. 

A \emph{strict partial order} over a set $U$ is a relation that satisfies irreflexivity, asymmetry, and transitivity. If, in addition, a strict partial order satisfies connectedness it is a \emph{strict total order}.
A \emph{total preorder} over a set $U$ is a relation that satisfies transitivity, reflexivity, and connectedness (i.e., for all two distinct $u, v\in U$, $u$ is in relation with $v$ or $v$ is in relation with $u$).
An \emph{equivalence relation} over a set $U$ is a relation that satisfies reflexivity, symmetry, and transitivity. For an equivalence relation $\sim$, we use $[u]_{\sim}$ to denote the equivalence class of an element $u\in U$ with respect to $\sim$, i.e., $[u]_{\sim} :=\{v\in U: u\sim v\}$. We denote with $U/_\sim$ the partition of $U$ consisting of all equivalence classes $[u]_{\sim}$ for $u\in U$. In this paper, we denote strict total orders with the symbols $\prec$ and $<$, total preorders with the symbols $\preceq$ and $\le$, and equivalence relations with the symbols $\sim$ and $\approx$.

A total preorder $\preceq$ over $U$ induces an equivalence relation $\sim$ 
over $U$: For $u,v\in U$, define $u\sim v$ if and only if $u\preceq v$ and $v\preceq u$. 
Throughout the paper, $\sim$ will always refer to the equivalence induced by $\preceq$ (the order $\preceq$ will always be unambiguously defined).
A total preorder $\preceq$, in addition, yields a strict total order $\prec$  
on the elements of $U/_\sim$ as $[u]_{\sim} \prec [v]_{\sim}$ if and only if $u\preceq v$ and not $v\preceq u$.
Throughout the paper $\prec$ will refer to the strict order induced by $\preceq$ (always unambiguously defined).
A total preorder $\preceq$ over a set $U$ can thus be represented by a unique ordered partition $\langle U_1,\ldots, U_k \rangle$, where the parts $U_i$ represent the equivalence classes with respect to $\sim$ and their ordering represents the strict total order $\prec$. 

\subparagraph{Non-Deterministic Finite Automata (NFAs).}
Let $\Sigma$ denote a fixed finite and non-empty \emph{alphabet} of \emph{letters}. We assume that there is a strict total order $<$ on the alphabet $\Sigma$. 
\begin{definition}[NFA and DFA]
    A \emph{non-deterministic finite automaton (NFA)} over the alphabet $\Sigma$ is an ordered triple $\AAA=(Q, \delta, s)$, where $Q$ is the set of \emph{states}, $\delta:Q\times \Sigma \rightarrow 2^Q$ is the \emph{transition function}, and $s\in Q$ is the \emph{source state}. 

    A \emph{deterministic finite automaton (DFA)} over the alphabet $\Sigma$ is an NFA over $\Sigma$ such that $|\delta(v,a)|\le 1$ for all $a\in \Sigma$ and $v\in V$.
\end{definition}
We note that the standard definition of NFAs includes also a set of final states. As we are not concerned with the accepting languages of automata in this work, we omit the final states from the definition. In what follows we consider the alphabet $\Sigma$ to be fixed and thus frequently refer to NFAs without specifying the alphabet. Given an NFA $\AAA = (Q, \delta, s)$, for a state $u\in Q$ and a letter $a\in \Sigma$, we use the shortcut $\delta_a(u)$ for $\delta(u,a)$, similarly $\delta^{-1}_a(u) = \{v\ :\ \delta(v,a)=u\}$. For a set $S \subseteq Q$, we let $\delta_{a}(S) := \bigcup_{u\in S} \delta_{a}(u)$ and $\delta^{-1}_{a}(S) := \bigcup_{u\in S} \delta^{-1}_{a}(u)$. 

The set of finite strings over $\Sigma$, denoted  by $\Sigma^*$, is the set of finite sequences of letters from $\Sigma$. 
We extend the transition function from letters to finite strings in the common way, i.e, for $\alpha\in \Sigma^*$ and $v\in Q$, we define $\delta_\alpha(v)$ recursively as follows. If $\alpha=\eps$, we let $\delta_\alpha(v)=\{v\}$. Otherwise, if $\alpha = a \alpha'$ for some $a\in \Sigma$ and $\alpha' \in \Sigma^*$, we let $\delta_\alpha(v) = \bigcup_{w\in \delta_a(v)} \delta_{\alpha'}(w)$. 

\begin{definition}[Strings Reaching a State] 
    Given an NFA $\AAA=(Q, \delta, s)$, for a state $v\in Q$, the \emph{set of strings reaching $v$} is defined as 
    $
        S_v = \{\alpha \in \Sigma^*: v\in \delta_\alpha(s)\}
    $.
\end{definition}

Equivalently, $S_v$ is the regular language recognized by state $v$.
For simplicity, in this work we assume that NFAs satisfy the following two properties: (1)~The source state has no in-going transitions, i.e., $S_s=\{\epsilon\}$ and every other state is reachable from the source state, i.e., $|S_v| \geq 1$ for all $v\in Q\setminus \{s\}$.
(2)~The automaton is \emph{input-consistent}, i.e., for every state $v\in Q\setminus \{s\}$ it holds that $\delta_a^{-1}(v)$ is non-empty for exactly one letter $a\in \Sigma$. 
It is easy to see that any automaton can be transformed into an equivalent one (in the sense of the recognized language) satisfying assumptions (1) and (2). Condition (1) yields that $|Q|\le |\delta| + 1$. Condition (2) comes at the cost of increasing the number of states and transitions by a factor of at most $|\Sigma|$. We also assume that a given NFA contains at least one transition for every letter.
We denote by $\lambda(v)$ the unique label of the in-going transitions of a state $v\in Q\setminus \{s\}$, while $\lambda(s):=\epsilon$. We also extend this notation to sets of states, that is for a set of states $S$, we let $\lambda(S) := \{\lambda(v) : v \in S\}$ and if $\lambda(S) = \{a\}$, we write $\lambda(S) = a$.

\subparagraph{Forward-Stable Partitions.}
Alanko et al.\ consider \emph{forward-stable partitions}~\cite[Section 4.2]{alanko:iac21}. 
\begin{definition}[Forward-Stability]
    Given an NFA $\AAA=(Q, \delta, s)$ and two sets of states $S,T\subseteq Q$, we say that $S$ is \emph{forward-stable} with respect to $T$, if, for all $a\in\Sigma$,
    $
        S\subseteq\delta_a(T)
    $ or 
    $
        S\cap \delta_a(T) = \emptyset
    $ holds.
    A partition $\PPP$ of $\AAA$'s states is \emph{forward-stable for $\AAA$}, if, for any two parts $S, T \in \PPP$, it holds that $S$ is forward-stable with respect to $T$.
\end{definition}
A direct consequence of forward-stability is given in the following lemma, i.e., all states in the same part of a forward-stable partition are reached by the exact same set of strings.
\begin{lemma}\label{lem: fs partition strings}
    Let $\PPP$ be a forward-stable partition for an NFA $\AAA=(Q, \delta, s)$ and assume that $u,v\in P$ for some part $P\in\PPP$. Then, $S_u=S_v$.
\end{lemma}
This property can be proven easily using the definition of forward-stability, see, e.g., \cite[Lemma~4.7]{alanko:iac21}.  Furthermore, there is a straightforward relationship between forward-stability and bisimulation, see, e.g., the work of Kanellakis and Smolka~\cite{KanellakisS90} or Chapter~7.3 of the book by Katoen and Baier~\cite{BaierK08}. The \emph{coarsest} forward-stable partition for an NFA $\AAA$, i.e., the forward-stable partition with fewest parts, is identical to the partition consisting of the equivalence classes with respect to the bisimilarity relation (the unique largest bisimulation) on $\AAA^{-1}$, the automaton obtained from $\AAA$ by reversing all its transitions. We include a proof of this fact in Appendix~\ref{app: bisimilarity}. This also directly yields that there is a unique forward-stable partition. We note that a reverse statement of Lemma~\ref{lem: fs partition strings} may not necessarily hold even for the coarsest forward-stable partition. More precisely, states in different parts of the coarsest forward-stable partition may have the same set of strings reaching them, see the left automaton in Figure~\ref{fig: example NFAs}.
In what follows, for an NFA $\AAA=(Q, \delta, s)$ and an equivalence relation $\sim$ on $Q$, we write $\AAA/_\sim$ for the quotient NFA (of $\AAA$ with respect to $\sim$) obtained by collapsing the equivalence classes into single states, see~\cite[Definition 8]{alanko:iac21} for a formal definition.

\subparagraph*{Wheeler NFAs and Quasi-Wheeler NFAs.}
Wheeler NFAs~\cite{gagie2017wheeler} are a special class of NFAs that can be stored compactly and indexed efficiently as they can be endowed with a specific type of strict total order, the so-called \emph{Wheeler order}.
\begin{definition}[Wheeler NFA and Wheeler order]
    Let $\AAA =(Q, \delta, s)$ be an NFA. We say that $\AAA$ is a \emph{Wheeler NFA}, if there exists a \emph{Wheeler order} $\prec$ of $Q$. A Wheeler order $\prec$ of $Q$ is a strict total order on $Q$ such that the source state precedes all other states, i.e., $s\prec v$ for all $v\in Q\setminus \{s\}$, and, for any pair $v \in \delta_a(u)$ and $v' \in \delta_{a'}(u')$:
    \begin{enumerate}[(1)]
        \item\label{first wheeler axiom} If $a < a'$, then $v \prec v'$.
        \item\label{second wheeler axiom} If $a = a'$, $u \prec u'$, and $v\neq v'$, then $v \prec v'$.
    \end{enumerate}
\end{definition}

Clearly, not every NFA is a Wheeler NFA and, even worse, recognizing if a given NFA is Wheeler is NP-complete (for alphabet size at least 2) as shown by Gibney and Thankachan~\cite{GibneyT19}. We now introduce the following problem, called \textsc{OrderWheeler}: Given an arbitrary NFA \(\AAA=(Q, \delta, s)\)
as input, the task is to compute a strict total order $\prec$ of $Q$ with the property that $\prec$ is a Wheeler order if $\AAA$ is Wheeler.
As a result of the NP-completeness of recognizing Wheeler NFAs previously mentioned, also \textsc{OrderWheeler} is NP-hard. This follows as checking whether a given order is indeed a Wheeler order can be done in linear time~\cite[Lemma 3.1]{alanko2020regular}.
This motivates introducing the following relaxed version of Wheeler NFAs.
\begin{definition}[Wheeler preorders and quasi-Wheeler NFAs]\label{definition: wheeler pre order}
    Let $\AAA =(Q, \delta, s)$ be an NFA. A Wheeler preorder $\preceq$ on $Q$ is a total preorder on $Q$ such that:
    \begin{itemize}
        \item The partition $Q/_\sim$, where $\sim$ is the equivalence relation induced by $\preceq$, is equal to the coarsest forward-stable partition of $\AAA$.
        \item The quotient automaton $\AAA/_{\sim}$ is a Wheeler NFA with the strict total order $\prec$ induced by $\preceq$ on the equivalence classes with respect to $\sim$.
    \end{itemize}
We say that $\AAA$ is a \emph{quasi-Wheeler NFA}, if there exists a \emph{Wheeler preorder} $\preceq$ on $Q$.
\end{definition}
See the beginning of this section for the formal definition of $\sim$ and $\prec$.
From the point of view of indexing, quasi-Wheeler NFAs are as useful as Wheeler NFAs: note that the quotient automaton $\AAA/_{\sim}$ in the definition above generates the same language as the original NFA $\AAA$ due to the properties of the forward-stable partition, see~\Cref{lem: fs partition strings}. However, $\AAA/_{\sim}$ is a Wheeler NFA and can thus be stored compactly and indexed efficiently.
We next note that Corollary 4.11 in the paper by Alanko et al.~\cite{alanko:iac21} directly yields the following lemma. 
\begin{lemma}
    Let $\AAA =(Q, \delta, s)$ be a Wheeler NFA, then $\AAA$ is also a quasi-Wheeler NFA.
\end{lemma}

Furthermore, the set of quasi-Wheeler NFAs is strictly larger than the set of Wheeler NFAs as there exist non-Wheeler NFAs $\AAA$ for which the quotient automaton $\AAA/_{\sim}$ is Wheeler, e.g., the NFA and its corresponding quotient NFA in the center and on the right of Figure~\ref{fig: example NFAs}.

\begin{figure}[!htb]
\centering
\resizebox{\textwidth}{!}{
\begin{tikzpicture}[->,>=stealth', semithick, auto, scale=.7]
    \node[state, label=above:{}] (S) at (-2,0){$s$};
    \node[state, label=above:{}] (Q_12) at (0,0) {$1$};
    \node[state, label=above:{}] (Q_3) at (2,0)	{$2$};
    \node[state, label=above:{}] (Q_4) at (4,0) {$3$};
    \draw (S) edge node[below] {$a$} (Q_12);
    \draw (S) edge [bend left=40, above] node {$a$} (Q_3);
    \draw (Q_12) edge node[below] {$a$} (Q_3);
    \draw (Q_3) edge   node {$a$} (Q_4);
    \draw (S) edge [bend right=35, above] node[below] {$a$} (Q_4);
    \draw (Q_3) edge  [loop above] node {$a$} (Q_3);
    
    \begin{scope}[xshift=8cm]
        \node[state, label=above:{}] (S)    at (-2,0)      {$s$};
        \node[state, label=above:{}] (Q_1)    at (0,1)    {$1$};
        \node[state, label=above:{}] (Q_2)    at (0,-1)      {$2$};
        \node[state, label=above:{}] (Q_3)    at (3,1)           {$3$};
        \node[state, label=above:{}] (Q_4)    at (3,-1)          {$4$};
        \draw (Q_1) edge node {$b$} (Q_3);
        \draw (Q_2) edge node[below] {$b$} (Q_4);
        \draw (Q_1) edge node[pos=0.3] {$b$} (Q_4);
        \draw (Q_2) edge node[pos=0.3, below] {$b$} (Q_3);
        \draw (S) edge [bend left=55, above] node {$b$} (Q_3);
        \draw (S) edge [bend left, above] node[below] {$a$} (Q_1);
        \draw (S) edge [bend right, above] node {$a$} (Q_2);
       \begin{scope}[xshift=7cm]
            \node[state, label=above:{}] (S)    at (-2,0)      {$s$};
            \node[state, label=above:{}] (Q_12)    at (0,0)      {$1,2$};
            \node[state, label=above:{}] (Q_3)    at (2,0)           {$3$};
            \node[state, label=above:{}] (Q_4)    at (4,0)           {$4$};
            \draw (S) edge node[below] {$a$} (Q_12);
            \draw (S) edge [bend left=40, above] node {$b$} (Q_3);
            \draw (Q_12) edge node {$b$} (Q_3);
            \draw (Q_12) edge [bend right=40, above] node[below] {$b$} (Q_4);
       \end{scope} 
    \end{scope}
\end{tikzpicture}
}
\caption{Left: An NFA for which $\PPP=\langle\{s\}, \{1\}, \{2\}, \{3\}\rangle$ is the coarsest forward-stable partition, yet $S_2=S_3$. Center: A non-Wheeler NFA (the two states $3$ and $4$ cannot be ordered) that is quasi-Wheeler as witnessed by the Wheeler preorder $\preceq$ corresponding to the ordered partition $\langle\{s\}, \{1, 2\}, \{3\}, \{4\} \rangle$. Right: Quotient automaton $\mathcal A /_\sim$, where $\mathcal A$ is the automaton in the middle and $\sim$ is the equivalence relation induced by the Wheeler preorder $\preceq$. Notice that $\mathcal A /_\sim$ is Wheeler with the strict total order $\prec$ induced by $\preceq$ on the equivalence classes with respect to $\sim$.}
\label{fig: example NFAs}
\end{figure}
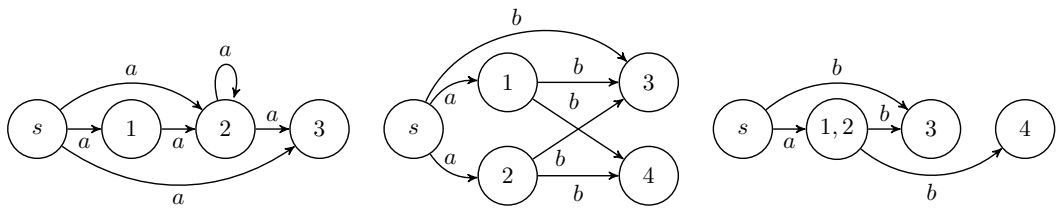

An important advantage of quasi-Wheeler NFAs over classical Wheeler NFAs is that the former can be recognized in polynomial time, which is implicit in the Forward Algorithm of Alanko et~al.~\cite{alanko:iac21}. Indeed, the Forward Algorithm receives a Wheeler NFA as input and, in $O(|\delta||Q|^2)$ time, outputs a Wheeler order of the quotient automaton $\AAA/_{\sim}$, where $\sim$ is the equivalence relation induced by the returned partition. However, this algorithm actually solves a slightly more general problem, since their output partition is guaranteed to correspond to a Wheeler order of $\AAA/_{\sim}$ whenever $\AAA$ is quasi-Wheeler rather than Wheeler. In other words Alanko et al.\ gave a polynomial time algorithm for the following problem that we call \textsc{PreOrderWheeler}: Given an arbitrary
NFA \(\AAA=(Q, \delta, s)\), the task is to compute a 
total preorder $\preceq$ of $Q$ with the property that $\preceq$ is a Wheeler preorder if $\AAA$ is quasi-Wheeler.
Their polynomial time algorithm for \textsc{PreOrderWheeler}, the \emph{Forward Algorithm}, yields a  polynomial-time recognition algorithm for quasi-Wheeler NFAs as follows: Given an arbitrary input NFA $\AAA$, run the Forward Algorithm to compute the Wheeler preorder $\preceq$ given by an ordering of the forward-stable partition. The output partition is guaranteed to be the coarsest forward-stable partition. Compute the quotient automaton with respect to this partition and check if it is Wheeler (in polynomial time by~\cite[Lemma 3.1]{alanko2020regular}) when endowed with the induced strict total order $\prec$. 

\section{Partition Refinement  for Wheeler Preorders of NFAs}
\label{sec: ordered partition refinement}
In this section we provide an algorithm that solves \textsc{PreOrderWheeler} in $O(|\delta| \log |Q|)$ time. Recall that the Forward Algorithm of Alanko et al.~\cite{alanko:iac21} has running time $O(|\delta|\cdot |Q|^2)$. We achieve the nearly-linear time complexity by using the partition refinement framework of Paige and Tarjan~\cite{PaigeT87}. It is in fact clear that this framework can be used to compute a forward-stable partition. Our notion of forward-stability corresponds to stability with respect to $|\Sigma|$ relations, defined as $E_a:=\{(u, v)\in Q^2: u\in \delta_a(v)\}$ for $a\in \Sigma$, in their terminology~\cite[Section 3]{PaigeT87}.
Our main contribution here is to extend the framework so as to compute an \emph{ordered} partition (and thus a preorder), while maintaining the nearly-linear running time.

We proceed with a description of the partition refinement algorithm by Paige and Tarjan.
The algorithm maintains two partitions $\PPP$ and $\XXX$, here $\PPP$ is an input partition to be refined and $\XXX$ is such that (1) $\PPP$ is a refinement of $\XXX$ and (2) every part of $\PPP$ is stable with respect to every part of $\XXX$. Initially, $\XXX$ is the partition with a single part containing all elements. Until $\PPP$ becomes stable, i.e., every part of $\PPP$ is stable with respect to every part of $\PPP$, which is witnessed by the fact that $\PPP=\XXX$, the algorithm iteratively chooses a \emph{compound part} $S$ from $\XXX$, i.e., a part form $\XXX$ that consists of multiple parts in $\PPP$. Then a ``splitter'' $B$ is chosen as one of the parts in $\PPP$ contained in $S$ and every part in $\PPP$ is refined using $B$ by doing a so-called ``three-way split'' that we detail in the next section. The idea behind the three-way split is essentially to make $\PPP$ stable with respect to $B$ and $S\setminus B$ at the same time. 
In fact, the choice of $B$ in the algorithm of Paige and Tarjan is crucial. It is fundamental to always choose a block $B$ such that $|B|\le |S|/2$. This is the essential property that yields the nearly linear running time using the observation that every element is contained in at most logarithmically many splitters. 
This property on the size of $B$ is achieved in Paige and Tarjan's implementation by inspecting the first two elements of the list of all parts from $\PPP$ contained in $S$ and then choosing $B$ to be the smaller one (in size). 

As it turns out, this choice of $B$ however interferes with maintaining the order of the parts such as to satisfy the properties of the Wheeler preorder. We solve this issue by choosing the smaller (in size) among the first and last block (in the sense of the ordered partition) contained in the first compound block (rather than the smaller out of the first two blocks, as in Paige and Tarjan's algorithm). This choice both guarantees that $|B|\le |S|/2$ and enables us to maintain a consistent ordering between the parts resulting from a split step.

\subparagraph*{Algorithm.}
We proceed with a description of our algorithm. A pseudo-code formulation can be found in Algorithm~\ref{alg: partition refinement}.
First, note that an input-consistent NFA has a natural partition of its states into the $|\Sigma| + 1 $ parts $\{Q_a\}_{a\in \Sigma \cup \{\eps\}}$, where $Q_a:= \{v\in Q: \lambda(v) = a\}$ for $a\in \Sigma \cup \{\eps\}$.  Property~\eqref{first wheeler axiom} of Wheeler orders already defines the ordering of these parts: any Wheeler preorder of the NFA's states has to satisfy that $u$ precedes $v$ if $u$'s in-coming letter is smaller than $v$'s. Hence, the ordering of the $|\Sigma| + 1$ parts has to be $Q_\eps, Q_{a_1}, \ldots Q_{a_k}$, where we assume that $\Sigma = \{a_1, \ldots, a_k\}$ with $a_1<\ldots<a_k$. Following this observation, the ordered partition $\PPP$ in the algorithm is initialized to $\PPP:=\langle Q_\eps, Q_{a_1}, \ldots Q_{a_k}\rangle$.

\begin{algorithm}[!ht]
\caption{Ordered Partition Refinement }\label{alg: partition refinement}
    \Input{NFA $\A=(Q,\delta, s)$}
    \Output{Ordered Partition $\PPP$ of $Q$ that corresponds to Wheeler preorder if and only if $\AAA$ is quasi-Wheeler}
  
    \medskip
    \tcp{* initialization * //}
    $\PPP := \langle Q_{\epsilon}, Q_{a_1}, ..., Q_{a_k} \rangle $,
    $\XXX := \langle Q\rangle$\; 

    \medskip
    \While{$\XXX\neq \PPP$
    }{
        \medskip
        \tcp{* get splitter $B$ * //}
        $S:=$ first block in $\XXX$ consisting of multiple blocks in $\PPP$\label{line: choose S}\\
        $B:=$ smaller of first and last block from $\PPP$ contained in $S$ \\\texttt{\hspace{23mm} // * variable $B$ contains this block even if split *//}\label{line: first or last}\;
        \medskip
        \tcp{* update $\XXX$ * //}
        \lIf{$B$ is first block}{replace $S$ in $\XXX$ with $B, S\setminus B$ \textbf{else} with $S\setminus B, B$}\label{line: update X}
    
        \medskip
        \tcp{* split blocks in $\PPP$ using $B$ * //}
        \For{$D\in\PPP$ }{
            \medskip
            \tcp{* split $D$ * //}

            $D_1:= D\cap \delta_a(B)$, 
            $D_2:= D\setminus D_1$, where $a= \lambda(D)$ \label{line: split D}\;
            $D_{11}:= D_1\cap \delta_a(S\setminus B)$,
            $D_{12}:= D_1\setminus D_{11}$\;
    
            \medskip
            \tcp{* update $\PPP$ * //}
            \lIf{$B$ first block}{
                replace $D$ in $\PPP$ with non-empty sets from $D_{12}, D_{11}, D_2$\\
            \textbf{else} with non-empty sets from $D_2, D_{11}, D_{12}$}

            \medskip
            \tcp{* continue \normalfont{\textbf{for}} \texttt{behind newly inserted blocks in $\PPP$ * //}}
        }    
    }
    \Return{$\PPP$}
\end{algorithm}

\begin{figure*}[ht!]
 	\centering
 		\includegraphics[width=0.4\textwidth, trim={2.0mm 1.3mm 2.0mm 2.0mm}, clip]{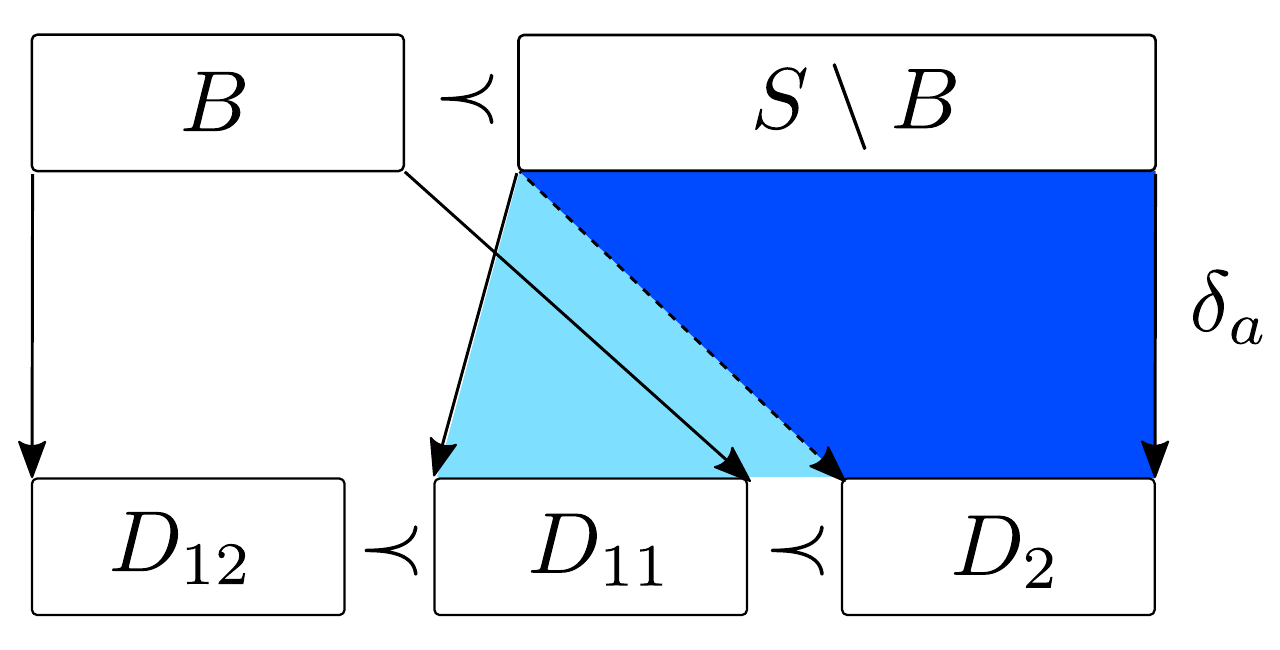}
        \hspace{5mm}
     	\includegraphics[width=0.4\textwidth, trim={2.0mm 1.3mm 2.0mm 2.0mm}, clip]{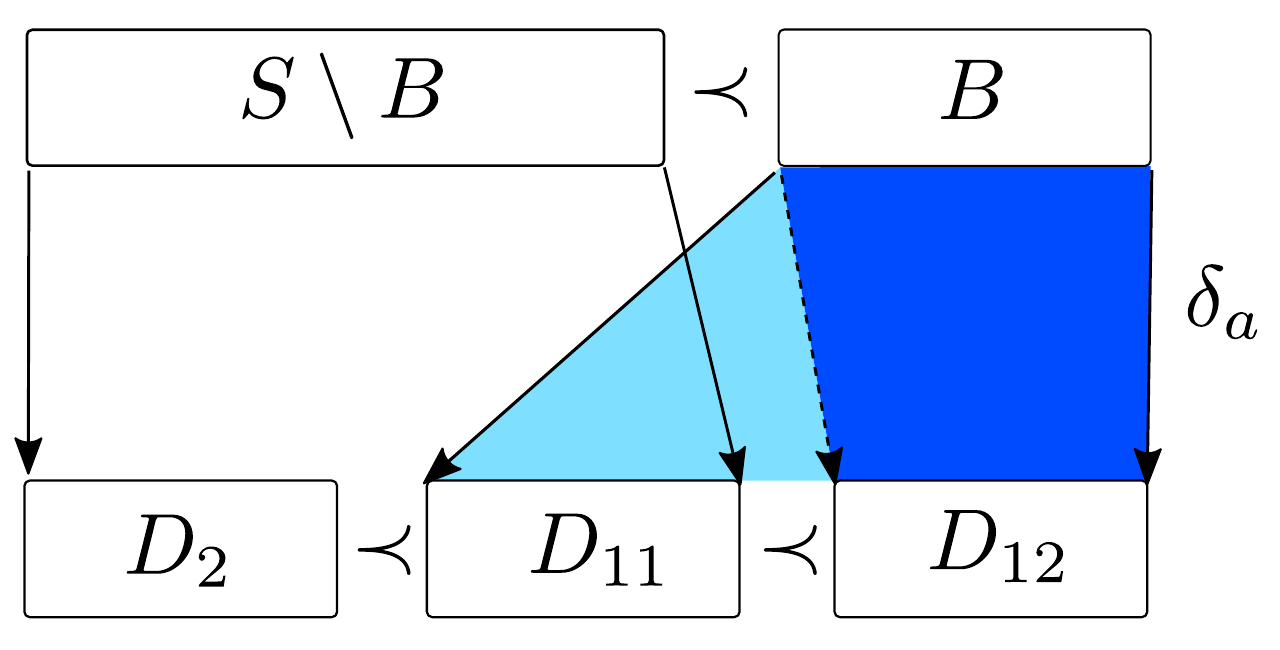}
    \caption{The two cases of our ``ordered'' three-way split of $D$ into $D_2$, $D_{11}$ and $D_{12}$ using splitter $B$ contained in the compound part $S$. Here, $D_{11}$ is the part of $D$ that is reached both by $B$ and $S\setminus B$, $D_{12}$ is the part of $D$ reached by $B$ but not by $S\setminus B$, and $D_2$ is the part of $D$ reached by $S\setminus B$ but not by $B$. The size of the splitter $B$ is at most half of the size of the compound block $S$ containing $B$. On the left, $B$ is the first part in $S$ (in the ordered partition), on the right it is the last part. On the left, the order of $B, S\setminus B$ gets propagated forward, the resulting order within $D$ is $D_{12}, D_{11}, D_{2}$. On the right, the order of $S\setminus B, B$ gets propagated forward, the resulting order within $D$ is $D_{2}, D_{11}, D_{12}$. The two shades of blue indicate the pruning of edges explained in \Cref{sec: DFA co-lex order}.}
    \label{fig:casesRefinement}
\end{figure*}

As in the original partition refinement  framework, the algorithm then keeps splitting the parts of the partition using so called \emph{splitters} which are themselves parts in the partition. By splitting, we mean the operation of making all parts of the partition forward-stable with respect to the splitter. As in Paige and Tarjan's algorithm, to maintain the set of splitters, the algorithm also maintains another ordered partition $\XXX$ such that $\PPP$ is a refinement of $\XXX$. Initially $\XXX$ has a unique part that is equal to $Q$ and the algorithm will maintain the invariant that $\PPP$ is forward-stable with respect to each part from $\XXX$. As before, we call the parts of $\XXX$ that contain more than a single part from $\XXX$ \emph{compound parts}. The algorithm then iteratively takes the first (in the order of the ordered partition) compound block $S$ from $\PPP$ and defines the splitter $B$ as the smaller (in size) of the first and last (in the order of the ordered partition) block in $\PPP$ contained in $S$.

Once the splitter $B$ is defined, the algorithm aims to split each part $D$ of the partition $\PPP$ (including $B$ itself) using $B$. As we aim at making the partition forward-stable with respect to $B$ we would want to split $D$ into $D_1:= D\cap \delta_a(B)$ and $D_2:= D\setminus D_1$. To implement the three-way split, we however want to further refine $D$ by $S\setminus B$. Recall that $S$ was a part in $\XXX$ that contained $B$ and that $\PPP$ is already forward-stable with respect to $S$ by the invariant. The forward-stability of $D$ with respect to $S$ yields that all states in $D_2$ are reached by $S\setminus B$ and thus $D$ is decomposed into only three parts 
$D_{11}:= D_1\cap \delta_a(S\setminus B)$,
$D_{12}:= D_1\setminus D_{11}$, and 
$D_2:= D\setminus D_1$
when splitting both with $B$ and $S\setminus B$.
This three-way split can be implemented with work proportional to the number of edges leaving $B$ (rather than $B$ and $S\setminus B$). Together with the choice of $B$ being the smaller out of the first and the last block contained in $S$, this is the main property that allows to prove the nearly linear running time bound of the algorithm. The order in which the sets $D_{11}$, $D_{12}$, and $D_2$ are put in $\PPP$ replacing $D$ depends on whether $B$ was the first or the last block in $S$ and is chosen so as to satisfy property~\ref{second wheeler axiom} of Wheeler orders. Intuitively, the order of $B$ and $S\setminus B$ is propagated forward to $D_{11}$, $D_{12}$, and $D_2$, see Figure~\ref{fig:casesRefinement} for an illustration of this ``ordered'' three-way split.

The algorithm keeps splitting the parts in $\PPP$ in this way until $\XXX=\PPP$, i.e., until there are no compound blocks left. Stopping at this point is correct by the invariant that the partition $\PPP$ is forward-stable with respect to each part from $\XXX=\PPP$, and, thus, forward-stable for $\AAA$.

\subparagraph*{Analysis.}
The following lemma states that the claimed invariant, that every part of $\PPP$ is forward-stable with respect to every part of $\XXX$, holds. This follows immediately from the same property of the framework by Paige and Tarjan (the proof of this lemma can be found in Appendix~\ref{sec: partition refinement proof}).
\begin{restatable}{lemma}{invariant}\label{lem: invariant}
    At the beginning of every iteration of the while loop in Algorithm~\ref{alg: partition refinement}, it holds that every part of the ordered partition $\PPP$ is forward-stable with respect to every part of $\XXX$.
\end{restatable}

We will now prove the core technical result of this section: The ordered partition refinement algorithm in fact computes a partition that corresponds to a Wheeler preorder.
\begin{lemma}
\label{lem: consistent wheeler}
    Assume that Algorithm~\ref{alg: partition refinement} is called on a quasi-Wheeler NFA $\A = (Q,s,\delta)$. Then, at any step of the algorithm, the partition $\PPP=\langle Q_1, \dots, Q_k\rangle$ agrees with every Wheeler order $\prec$ of the quotient automaton $\AAA':=\AAA/_{\PPP'}$, where $\PPP'$ denotes the coarsest forward-stable partition for $\AAA$. That is, if $i<j, u\in P_i, v\in P_j$ then $u\prec v$ for every Wheeler order $\prec$ of $\AAA'$.
\end{lemma}
\begin{proof}
    The initial partition $\PPP=\langle Q_\epsilon, Q_{a_1}, \dots, Q_{a_k}\rangle$ agrees with any Wheeler order $\prec$ of $\AAA'$. We will show by induction over the number of steps of the while loop that this is the case in any step of the algorithm. Assume that this is true before some intermediate iteration and let us denote with $\PPP$ the ordered partition at that point. Now, let $S$ be the compound block and let $B$ be the splitter chosen in that iteration. Let us call $\PPP'$ the ordered partition after refining $\PPP$ with $B$ and $S\setminus B$, i.e., at the end of that iteration. We will prove that $\PPP'$ agrees with any Wheeler order $\prec$. As in the algorithm, for $D\in\PPP$, assume that $a=\lambda(D)$ and let $D_{1}=D\cap \delta_a(B)$, $D_2=D\setminus D_1$, and $D_{11}= D_1\cap \delta_a(S\setminus B)$ and $D_{12}= D_1\setminus D_{11}$. Now let $u,v\in D$. If $u$ and $v$ are contained in the same set out of the three sets $D_{12}, D_{11}, D_2$, nothing is to be shown. Hence, assume that $u$ and $v$ are in two different sets out of the three sets. There are three cases: (1) $u\in D_{11}$ and $v\in D_2$, (2) $u\in D_{12}$ and $v\in D_2$, and (3) $u\in D_{12}$ and $v\in D_{11}$. Notice first that in all three cases, there exists $x\in B$ with $u\in \delta_a(x)$. Now, recall that $D$ is stable with respect to $S$ due to Lemma~\ref{lem: invariant} and thus there exists $y\in S\setminus B$ with $v\in \delta_a(y)$ in all three cases as well. Let us now first assume that $B$ is the first block in $S$ in line~\ref{line: first or last}. Then, as $B$ precedes all blocks in $S\setminus B$ in $\PPP$ and $\PPP$ is consistent with any Wheeler order of $\AAA'$, we have that $x\prec y$ for any Wheeler order of $\AAA'$. Then, property~\ref{second wheeler axiom} of Wheeler orders implies that $u\prec v$ for any Wheeler order of $\AAA'$. In summary, as the algorithm replaces $D$ by $D_{12}, D_{11}, D_2$ in this case, we deduce that $\PPP'$ again agrees with each Wheeler order of $\AAA'$. 
    If, instead $B$ is the last block in $S$ in line~\ref{line: first or last}, then $B$ succeeds all blocks in $S\setminus B$ in $\PPP$ and thus we can analogously deduce $v\prec u$ from property~\ref{second wheeler axiom} of Wheeler orders and as the algorithm replaces $D$ by $D_{2}, D_{11}, D_{12}$, we deduce that $\PPP'$ agrees with every Wheeler order of $\AAA'$ also in this case.
\end{proof}

It remains to argue that the ordered partition refinement algorithm in fact runs in the claimed running time bound of $O(|\delta| \log |Q|)$. 
The two main ingredients are as follows: (1) We always choose the smaller out of the first and the last part contained in $S$ as $B$ and thus $|B|\le |S|/2$. As a result every state is in at most logarithmically many splitters. (2) We can implement the algorithm in such a way that the work done in a refinement step with splitter $B$ is proportional to the size of $B$ and the number of out-going transitions from $B$. This can be argued completely analogously as done by Paige and Tarjan. We prove our algorithm running time and describe the main data structure details in Appendix~\ref{sec: eff implementation}. In summary, we obtain the following theorem.
\wheelerPreorder*

\section{Partition Refinement  for Width-Optimal Co-lex Orders of DFAs}\label{sec: DFA co-lex order}

Cotumaccio and Prezza~\cite{CotumaccioP21} propose a way of sidestepping the fact that there may not be a Wheeler order for a given NFA. They show how to index general NFAs using \emph{partial orders}. For this purpose, they define co-lex orders for NFAs as follows.

\begin{definition}[Co-lex Order] \label{def: colex order}
    Let $\AAA = (Q, \delta, s)$ be an NFA.
	A \emph{co-lex order} for $\AAA$ is a strict partial order $\prec$ of $Q$ such that the source state precedes all other states, i.e., $s\prec v$ for all $v\in Q\setminus \{s\}$, and, for any pair $v\in \delta_a(u)$ and $v'\in \delta_{a'}(u')$: 
    \begin{enumerate}[(1)]
        \item\label{first colex axiom} If $a < a'$, then $v \prec v'$. 
        \item\label{second colex axiom} If $a = a'$, $v \prec v'$, and $u\neq u'$, then $u \prec u'$.
    \end{enumerate}
\end{definition}

The \emph{width} of a strict partial order $\prec$ is defined as the size of the largest antichain, i.e. set of pairwise $\emph{incomparable}$ states, where two distinct states $u,v\in Q$ are said to be incomparable if neither $u\prec v$ nor $v\prec u$ holds. Finding a smallest-width co-lex order of $\AAA$ has been of particular interest for constructing efficient indexes for pattern matching on automata~\cite{CotumaccioP21,Cotumaccio:dcc22,CotumaccioDGPP:arxiv22,KimOP23}. We provide more context on co-lex orders in Appendix~\ref{sec: app: DFA: motivation}.
For DFAs, the following order is known to be the smallest-width co-lex order~\cite{CotumaccioP21}.
\begin{definition}\label{def: prec_DFA} 
    Let $\AAA=(Q, \delta, s)$ be a DFA. The relation $\prec_{\AAA}$ over $Q$ is defined as follows. For $u, v\in Q$, 
	$u \prec_{\AAA} v$ holds if and only if $\alpha < \beta$ for all $\alpha \in S_u$ and $\beta \in S_v$.
\end{definition}
In the above definition, the alphabet order $<$ over $\Sigma$ is extended to the co-lexicographical order on strings.
For the purpose of the co-lex order $\prec_\AAA$, a state $v\in Q$ can thus be represented solely using upper and lower bounds on $S_v$. In particular, for a state $v\in Q$, let $\inf S_v$ and $\sup S_v$ be the greatest lower bound (infimum) and the least upper bound (supremum) of $S_v$, respectively, which are possibly left-infinite strings. A formal definition of these concepts can be found in Appendix~\ref{sec: app: DFA: infima suprema}.
As shown by Kim et al.~\cite{KimOP23}, the co-lex order $\prec_\AAA$ can be characterized using the co-lex order of $\IS_\AAA:=\{\inf S_v:v\in Q\}\cup\{\sup S_v:v\in Q\}$. 
\begin{lemma}[\text{\cite[Theorem 10]{KimOP23}}]\label{lem: colex characterization}
    Let $\AAA = (Q, \delta, s) $ be a DFA. Then, for any two distinct states $u,v\in Q$,
    $u \prec_{\AAA} v$ holds if and only if 
    $\sup{S_u}\le \inf{S_v}$.
\end{lemma}

We note that, for the purpose of indexing, we also need a minimum chain partition of $\prec_\AAA$, i.e., a minimum-size partition of $Q$ such that the states in each part are totally ordered under $\prec_\AAA$. As shown by Kim et al.~\cite{KimOP23}, such a chain partition can be computed in linear time via a greedy algorithm for the interval graph coloring problem, provided that the set $\IS_\AAA$ (implicitly defining an interval in co-lex order for each state) has been computed and sorted.

In this section, we give an algorithm for co-lex sorting the set $\IS_\AAA$ that runs in nearly-linear time $O(|\delta| \log |Q|)$ on any input-consistent DFA. This is a significant improvement with respect to the best previously-known algorithms that run in $O(|\delta|\cdot|Q|)$ and $O(|Q|^2\log|Q|)$ time~\cite{KimOP23}.

\subparagraph{Outline.} Given an input DFA $\AAA=(Q,\delta,s)$, we first compute a pruned automaton $\AAA^{\inf}=(Q,\delta^{\inf},s)$ that \emph{encodes} $\{\inf S_v:v\in Q\}$ in a sense that walking back starting at a state $v\in Q$ on the pruned transition $\delta^{\inf}$ yields its infimum string $\inf S_v$. Similarly, we compute a pruned automaton $\AAA^{\sup}=(Q,\delta^{\sup},s)$ that encodes the suprema strings. As we will see, computing these sets takes $O(|\delta|\log|Q|)$ time each. On these two pruned automata, we compute the co-lex order of $\IS_\AAA$ in $O(|Q|\log|Q|)$ time in a similar way as done by Kim et al.~\cite{KimOP23}, from which the smallest-width co-lex order of $\AAA$ (with a minimum chain partition) can be computed in linear time.

\subparagraph{Pruning Algorithm.}
The algorithm is identical to Algorithm~\ref{alg: partition refinement} with the only difference that we insert an additional pruning step, Algorithm~\ref{alg: pruning}, before 
Line~\ref{line: split D} of the algorithm. See Figure~\ref{fig:casesRefinement} for an illustration: If a state in $D$ is reached from both $B$ and $S\setminus B$, we only keep the edges from the smaller (in the sense of the ordered partition) of the two. This corresponds to changing $\delta_a$ such that the blue portion of the figure shrinks to the dark blue one.

{\LinesNumberedHidden
\begin{algorithm}[!ht]
\caption{Pruning Step for Algorithm~\ref{alg: partition refinement} inserted before  Line~\ref{line: split D}}\label{alg: pruning}
    \medskip
    \tcp{* prune $\AAA$ * //}
    \lIf{$B$ is the first block}{
    delete transitions from $S\setminus B$ to $D_{11}$
    \textbf{else}
    from $B$ to $D_{11}$.
    }
\end{algorithm}}
We inductively show that Algorithm~\ref{alg: pruning} outputs a pruned automaton encoding the set of infima $I:=\{\inf S_v: v\in Q\}$. For an integer $i\ge 1$, let $\delta^{(i)}$ be the pruned transition at the end of the $i$-th iteration of the while loop. Let us analogously denote with $\XXX^{(i)}$ and $\PPP^{(i)}$ the state of the ordered partitions $\XXX$ and $\PPP$ in the algorithm at the end of the $i$-th iteration. For convenience, we also define $\delta^{(0)}=\delta$, $\XXX^{(0)}=\langle Q\rangle$, and $\PPP^{(0)}=\langle Q_\eps, Q_{a_1}, \ldots, Q_{a_k}\rangle$. 
Lemma~\ref{lem: invariant} and the definition of the pruning step yield the following invariant on forward-stability.
\begin{observation}\label{obs: invariant}
    At the end of every iteration $i$ of the while loop in Algorithm~\ref{alg: partition refinement} when run together with the pruning step, it holds that every part of the ordered partition $\PPP^{(i)}$ is forward-stable with respect to every part in $\XXX^{(i)}$ in the automaton $\AAA^{(i)}=(Q, \delta^{(i)}, s)$.
\end{observation}
Intuitively speaking, the pruning step removes transitions that do not originate from the co-lexicographically smallest part in the sorted partition. Hence, we obtain the following lemma (the proof is deferred to Appendix~\ref{sec: app: DFA: proofs}).
\begin{restatable}{lemma}{LemAAprime}\label{lem: AAprime}
    Let $v\in Q$ be a state with $a=\lambda(v)$ and let $A,A'\in\XXX^{(i)}$ for some $i\ge 0$ be such that $v\in\delta^{(i)}_a(A)\cap\delta_a(A')$. Then either (i) $A=A'$ or (ii) $A$ precedes $A'$ in $\XXX^{(i)}$.
\end{restatable}
Let $\PPP^*$ and $\delta^*$ be the ordered partition and the pruned transition, respectively, obtained at the end of the algorithm's execution. 
From Lemma~\ref{lem: AAprime}, we can 
see that for every state $v\in Q\setminus\{s\}$, there exists a unique part $A\in\PPP^{*}$ such that $v\in \delta_{\lambda(v)}^{*}(A)$. Combining this with Observation~\ref{obs: invariant}, we obtain that a unique string can be obtained by walking backwards through the pruned transitions $\delta^*$. For $u\in Q$, let $\alpha^*_u$ be the longest (or possibly left-infinite) string that can be obtained in this way starting from $u$. Since every transition comes from a co-lex smallest part, we can obtain the following lemma (the proof is included in Appendix~\ref{sec: app: DFA: proofs}).

\begin{restatable}{lemma}{LemInfimaEncoding}\label{lem: infima encoding}
    For every $u\in Q$, $\alpha^*_u=\inf S_u$.
\end{restatable}

It is worth noting that some states possibly have more than one in-going transitions after the termination of the algorithm. Nevertheless, Lemma~\ref{lem: infima encoding} still holds when we choose any of them and remove the others. From this observation, we can assume that every state (except the source state) in the pruned automaton $\AAA^{\inf}$ has exactly one in-going transition.

Similarly, we can compute the pruned automaton for the set of suprema $\{\sup S_v: v\in Q\}$. The only difference is that we start with the partition $\PPP'^{(0)}=\langle Q_{a_k}, \ldots, Q_{a_1}, Q_\eps\rangle$ with the reversed order. It is simple to see the greatest string can be computed with this setting.

Regarding the time complexity, notice that the partition refinement algorithm iterates through the same partitions as if it were run on $\AAA^{\inf}$ in the first place. Additional time is taken for the pruning step for deleting transitions. This work can be done in $O(1)$ time per transition. Once a transition is deleted, it will never be considered in the rest of the execution and hence the additional work amortizes to $O(|\delta|)$. Consequently, the asymptotic time complexity of the partition refinement algorithm with pruning remains $O(|\delta|\log|Q|)$.

\subparagraph{Computing Co-Lex Order of $\IS_\AAA$.} 
Once the two pruned automata $\AAA^{\inf}$ and $\AAA^{\sup}$ are obtained, we can easily compute the co-lex order of $\IS_\AAA$ in $O(|Q|\log|Q|)$ time using the suffix doubling algorithm~\cite{KimOP23}, which extends the well-known prefix-doubling algorithm~\cite{ManberM90}. Instead of accessing via integer indexes for the doubling procedure, each state keeps a pointer referring to another (possibly the same) state that is $2^k$ hops away along its backward walk. We describe all details of this algorithm in Appendix~\ref{sec: app: DFA: suffix doubling}.
In summary, we conclude this section with the following theorem.
\wheelerPoset*

\section{Experimental results}
We implemented our partition refinement algorithms in \texttt{C++} and made it available at \url{https://github.com/regindex/finite-automata-partition-refinement.git}. We compared the algorithm from \Cref{theorem: wheeler preorder} to the only (to the best of our knowledge) other available tool for computing a sorted partition consistent with any Wheeler order of the input NFA, the renaming heuristic from \texttt{WGT}~\cite{chao2022wgt}. For a fair comparison, we only run the first part of the \texttt{WGT} recognizer and remove the exponential time search that is used to subsequently compute a Wheeler order of the states ending up in equivalence classes.\footnote{As observed empirically, the partition computed by \texttt{WGT} is in some cases coarser than the partition computed by our algorithm and is, thus, not necessarily forward-stable.} We used the Wheeler graph generator included in \texttt{WGT} to generate $7$ random input Wheeler NFAs with $|Q| \in \{5^6 \cdot 2^{i}: i=0,\ldots ,6 \}$ and $|\delta| = 3|Q|$ (we cannot generate much denser graphs since for Wheeler graphs $\delta=O(|\Sigma|\cdot|Q|)$ and here $|\Sigma|=5$). Our experiments were run on a server with Intel(R) Xeon(R) W-2245 CPU @ 3.90GHz with 8 cores and 128 gigabytes of RAM running Ubuntu 18.04 LTS 64-bit.
\begin{figure*}[ht!]
    \centering
 	\includegraphics[width=0.44\textwidth, trim={5.5mm 5.5mm 5.0mm 5.5mm}, clip]{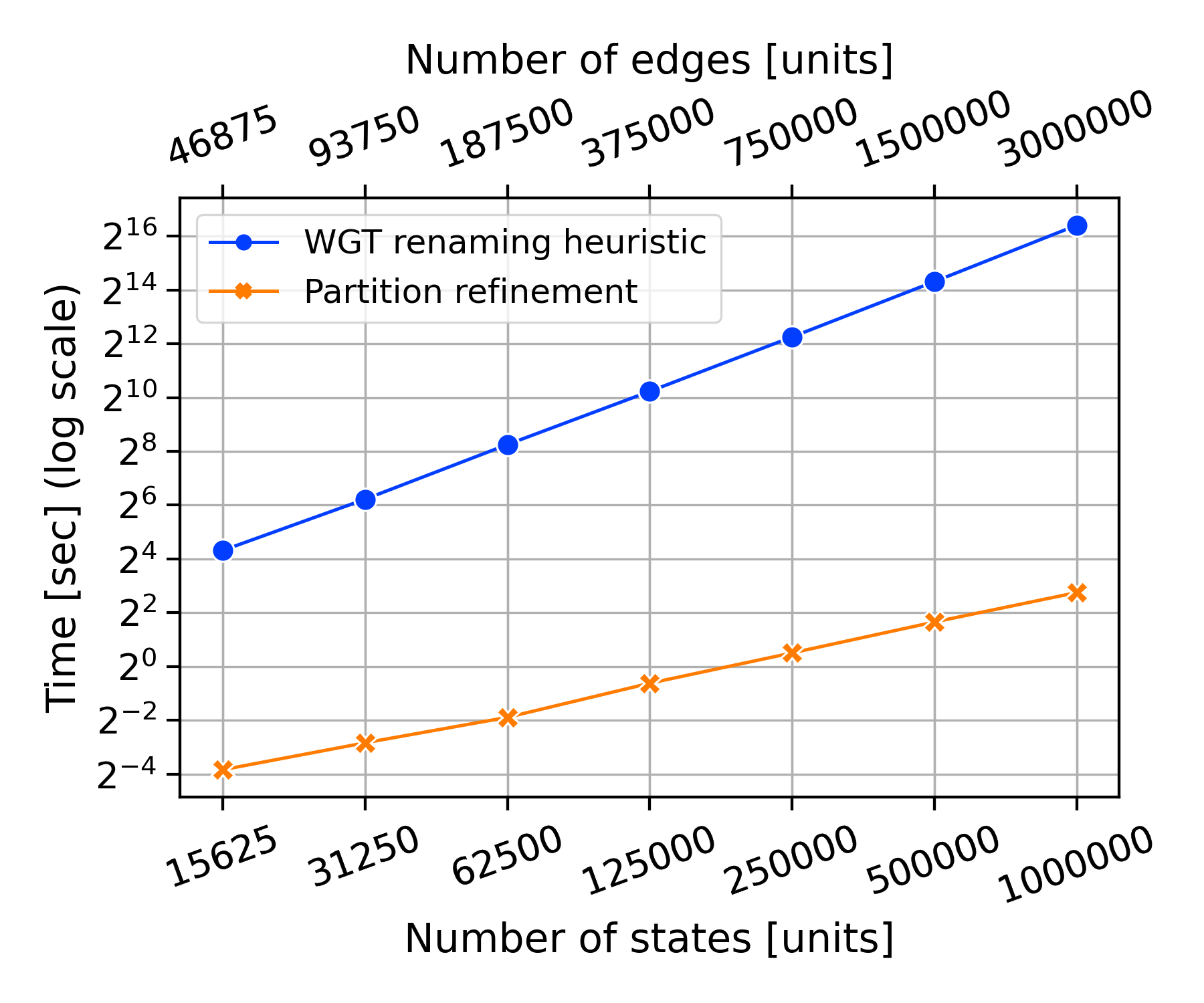}
    \hspace{3mm}
    \includegraphics[width=0.44\textwidth, trim={5.5mm 5.5mm 5.0mm 5.5mm}, clip]{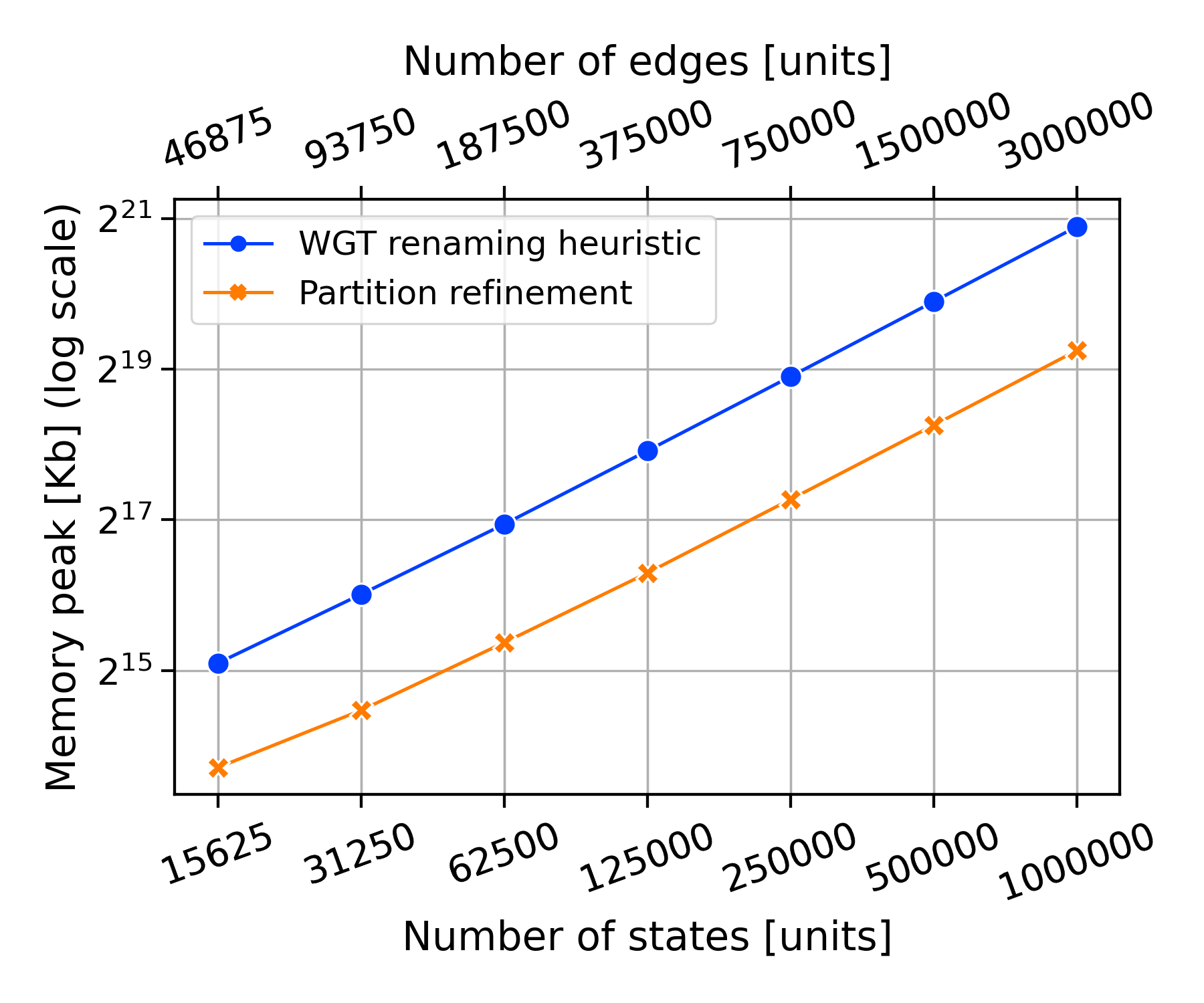}
    \caption{CPU time (left) and memory peak (right) for sorting seven Wheeler NFAs using our partition refinement algorithm and the renaming heuristic contained in the \texttt{WGT} recognizer software. Datasets generated using \texttt{WGT} specifying seven different combinations of number of states and edges.}
    \label{fig:exp1}
\end{figure*}

\Cref{fig:exp1} shows the running time and peak memory consumption of both implementations. As expected, our partition refinement implementation shows a slight super-linear behavior confirming the $O(|\delta|\cdot\log{|Q|})$ worst-case running time of the algorithm. On the other hand, \texttt{WGT} shows a quadratic behavior, which is better than the cubic $O(|\delta|\cdot|Q|^2)$ bound for the Forward Algorithm of Alanko et.~al~\cite{alanko:iac21}, but between $\approx 100\times$ to $\approx 10000\times$ slower than our implementation. In terms of peak memory, both implementations behave linearly with respect to the automaton's size with a $3\times$ advantage in favor of our implementation. On the largest input instance containing one million states, our implementation computes the Wheeler preorder in about seven seconds. In a second experiment, we show that our implementation of the algorithm from \Cref{theorem: wheeler poset} can prune a pangenomic DFA containing over 51 million states and 53 million edges in 355 seconds with a peak memory of 33.2 Gb.



\bibliography{main}

\newpage

\appendix

\section{Forward-Stability and Bisimulation}\label{app: bisimilarity}
\begin{definition}[Bisimulation]
	Given an NFA $\AAA=(Q, \delta, s)$ over a finite alphabet $\Sigma$, a bisimulation on $\AAA$ is a relation $B\subseteq Q \times Q$ such that for all $(u, v)\in B$ and $a\in \Sigma$, it holds that:
	\begin{itemize}
	    \item If $u'\in\delta_a(u)$, then there exists $v'\in \delta_a(v)$ such that $(u', v')\in B$.
        \item If $v'\in\delta_a(v)$, then there exists $u'\in \delta_a(u)$ such that $(u', v')\in B$.
    \end{itemize} 
\end{definition}
Any bisimulation $B$ on $\AAA$ partitions $V$ into sets $P$ that correspond to the weakly connected components of the directed graph $(V,B)$. We  denote the resulting partition of $V$ with $\PPP_B$. Note that the relation $\bar B:=\bigcup_{P\in\PPP_B}P^2$ is again a bisimulation that we call the \emph{closure} of $B$.

\begin{definition}[Bisimilarity]
	Given an NFA $\AAA=(Q, \delta, s)$ over a finite alphabet $\Sigma$, the \emph{bisimilarity} of $\AAA$ is the union of all bisimulations on $\AAA$. 
\end{definition}
It is clear that the bisimilarity of $\AAA$ is again a bisimulation on $\AAA$ and it partitions $Q$ into equivalence classes, we write $\PPP_\sim$ for the \emph{bisimilarity partition} on $G$ consisting of the partition with the equivalence classes as its parts.

Given an NFA $\AAA=(Q, \delta, s)$ over a finite alphabet $\Sigma$, we define the NFA $\AAA^{-1}=(Q, \delta^{-1}, s)$ as the automaton with the same state set $Q$ and the transition function $\delta^{-1}$ that is defined as $\delta^{-1}(u, a):=\delta^{-1}_a(u)$ for all $u\in Q$ and $a\in \Sigma$.

\begin{lemma}\label{lem: bisim -> forward-stable}
	Given an NFA $\AAA=(Q, \delta, s)$ over a finite alphabet $\Sigma$ and any bisimulation $B$ on $\AAA^{-1}$. Then, the partition $\PPP_{B}$ is forward-stable on $\AAA$.
\end{lemma}
\begin{proof}
	Let $S,T$ be two parts of $\PPP_{B}$. We show that $S\subseteq\delta_a(T)$ or $S\cap \delta_a(T)  = \emptyset$. Assume for contradiction that this is not the case, i.e., that $S\cap\delta_a(T)\neq \emptyset$ and $S\setminus \delta_a(T)\neq \emptyset$ holds. Hence, there exists $u\in S\cap\delta_a(T)$ and $v\in S\setminus \delta_a(T)$. As both $u\in S$ and $v\in S$, it follows that $(u,v)\in \bar B$. Notice however that $u\in S\cap\delta_a(T)$ and thus there exists $u'\in \delta_a^{-1}(u)\cap T$. Hence by the bisimulation $\bar B$, there exists $v'\in \delta_a^{-1}(v)$ such that $(u',v')\in \bar B$. As a consequence also $v'\in T$, a contradiction to $v\notin \delta_a(T)$. 
\end{proof}

\begin{lemma}\label{lem: forward-stable -> bisim}
	Given an NFA $\AAA=(Q, \delta, s)$ over a finite alphabet $\Sigma$, any forward-stable partition $\PPP$ on $G$ induces a bisimulation $B_\PPP:=\{(u,v)\in Q\times Q: u, v\in S, S\in \PPP\}$.
\end{lemma}
\begin{proof}
    We just need to check that $B_\PPP$ indeed is a bisimulation. Let $(u,v)\in B_\PPP$, then by definition $u, v\in S$ for some part $S\in \PPP$. Now assume that there exists $u'\in\delta^{-1}_a(u)$ for some $a\in\Sigma$. Let $T$ be the part of $\PPP$ such that $u'\in T$. By forward-stability it follows that $S\subseteq\delta_a(T)$ and thus there exists $v'\in T$ such that $v'\in \delta^{-1}_a(v)$. By the definition of $B_\PPP$ it furthermore follows that $(u', v')\in B_\PPP$. The second condition for $B_\PPP$ being a bisimulation is shown in a symmetric way.
\end{proof}

\begin{corollary}
    Given an NFA $\AAA=(Q, \delta, s)$ over a finite alphabet $\Sigma$, the coarsest forward-stable partition $\QQQ$ on $\AAA$ is identical to the bisimilarity partition $\PPP_\sim$ on $\AAA^{-1}$.
\end{corollary}
\begin{proof}
    Assume for contradiction that $\QQQ\neq \PPP_\sim$. Recall that $\PPP_\sim$ is forward-stable due to Lemma~\ref{lem: bisim -> forward-stable} and thus $|\QQQ|\le |\PPP_\sim|$ as $\QQQ$ is the coarsest forward-stable partition. Now, $\QQQ\neq \PPP_\sim$ and $|\QQQ|\le |\PPP_\sim|$ imply that there exist $u,v\in S\in\QQQ$ such that $u\in P_1$ and $v\in P_2$ with $P_1,P_2\in \PPP_\sim$ and $P_1\neq P_2$. By Lemma~\ref{lem: forward-stable -> bisim}, it follows that $B_\QQQ$ is a bisimulation and $(u,v)\in B_\QQQ$. As $\PPP_\sim$ is defined as the union of all bisimulations this is however a contradiction to $u$ and $v$ being in different parts in $\PPP_\sim$.
\end{proof}

\section{Deferred Material for Section~\ref{sec: ordered partition refinement}}
\subsection{Deferred Proofs}
\label{sec: partition refinement proof}
\invariant*
\begin{proof}
    Note that the property holds before the first iteration as $\XXX$ contains $Q$ as its unique part then and every state is assumed to have at least one in-going transition. Inductively assume that the property also holds before some intermediate iteration of the while loop. Assume that $S$ is the block from $\XXX$ chosen in that iteration. The only modification made to $\XXX$ is the replacement of $S$ with $B, S\setminus B$ or $S\setminus B, B$ in line~\ref{line: update X}. As we have to show that every part of $\PPP$ is stable with respect to every part of $\XXX$ at the end of the iteration. For the blocks $B$ and $S\setminus B$, this is clear by the refinement step of the algorithm. For the other blocks in $\XXX$ this follows from the induction hypothesis and the fact that every subset of a forward-stable set is again forward-stable.
\end{proof}

\begin{lemma}\label{lem: nearly linear}
    The following two hold:
    \begin{enumerate}[(1)]
        \item \label{line: refine bound} A refinement step with splitter $B$ takes $O(|B|+ \sum_{a\in\Sigma}\sum_{v\in B} |\delta_a(v)|)$ time. 
        \item \label{line: no of splitters} Every $v\in Q$ is in at most $\log |Q| + 1$ different splitters $B$ used for the refinement step.
    \end{enumerate}
    Hence, the ordered partition refinement  algorithm can be implemented in $O(|\delta|\cdot\log |Q|)$.
\end{lemma}
\begin{proof}
    Part~\eqref{line: refine bound} can be achieved by an efficient implementation that stores the partitions $\PPP$ and $\XXX$ in a specific data structure (based on double-linked lists) that is identical to the one described by Paige and Tarjan. We include more details in the following subsection. The data structure guarantees that every refinement step takes $O(1)$ per state contained in the splitter $B$ and $O(1)$ per transition leaving $B$, i.e., the refinement by splitter $B$ takes $O(|B|+ \sum_{a\in \Sigma}\sum_{v\in B} |\delta_a(v)|)$. For part~\eqref{line: no of splitters} notice that the splitter sets, say, $B_1,\ldots, B_\ell$ in which a fixed state $v$ is contained satisfy $|B_i|\ge |B_{i+1}|/2$ for each $i\in [\ell - 1]$. As $|B_1|\le |Q|$, it follows that $v$ is in at most $\ell \le \log |Q| + 1$.
    Let us now call $\BBB$ the set of all splitters used by the algorithm. The total time needed for all refinement steps can thus be bounded by 
    \begin{align*}
        \sum_{B\in\BBB} &O\Big(|B|+ \sum_{v\in B}\sum_{a\in\Sigma} |\delta_a(v)|\Big) 
        = \sum_{v\in Q} \sum_{B\in \BBB : v\in B} O(1) 
            + O\Big(\sum_{v\in Q} \sum_{B\in\BBB : v\in B} \sum_{a\in\Sigma} |\delta_a(v)|\Big)\\
        &= O(|Q| \log |Q|)
            + O\Big( \log |Q| \cdot \sum_{v\in Q} \sum_{a\in\Sigma} |\delta_a(v)|\Big)
        = O(|\delta| \cdot \log |Q|).\qedhere
    \end{align*}
\end{proof}

\subsection{Data Structure Details}
\label{sec: eff implementation}
The nearly linear time implementation can be achieved in the same way as was described already by Paige and Tarjan~\cite{PaigeT87}. We briefly repeat the most important parts of their implementation.
Each element $u\in Q$ is represented by a record. Each pair $u,v\in Q$ such that $v\in \delta_a(u)$ is represented by a record, which we call an \emph{edge} (labeled $a$). Each part in partition $\PPP$ and in partition $\mathcal{X}$ is represented by a record. A part $S\in\mathcal{X}$ is called \emph{simple} if it contains only a single part of $\PPP$ and \emph{compound} otherwise. 
Each edge $u,v \in Q$ s.t.\ $v\in\delta_a(u)$ points to element $u$. Each element $v$ points to a list of the edges $u,v$ s.t.\ $v\in\delta_a(u)$, which allows scanning $\delta_a^{-1}(v)$ in time proportional to its size. Each part of $\PPP$ has an associated integer equal to its size, and points to a doubly linked list of the elements in it (which allows for constant deletion time). Each element points to the part in $\PPP$ that contains it. Each part in $\mathcal{X}$ points to a doubly linked list of the parts of $\PPP$ in it. Each part of $\PPP$ points to the part of $\mathcal{X}$ that contains it. Additionally, $\mathcal{C}$ is a set of compound parts of $\mathcal{X}$, which is initialized to $Q$ and maintained accordingly.

For the three-way splitting, for every part $S\in\mathcal{X}$ and every element $x\in \delta_a(S)$, we maintain a record containing $\texttt{count}_a(x,S)=|S\cap \delta_a^{-1}(x)|$, and $x$ points to $\texttt{count}_a(x,S)$. Each edge $y,x$ labeled $a$ (i.e., $x\in \delta_a(y)$) such that $y\in S$ points to $\texttt{count}_a(x,S)$. Initially, $x$ points to $\texttt{count}_a(x,Q)=|\delta_a^{-1}(x)|$ and each edge $y,x$ labeled $a$ also points to $\texttt{count}_a(x,Q)$.

The space needed for all the data structures is $O(|\delta|)$, the total number of transitions. The initialization time is also $O(|\delta|)$. The algorithm stops when $\mathcal{C}=\emptyset$ or equivalently $\PPP=\mathcal{X}$.
In what follows, we specify the refinement step in more detail.

\subparagraph{Implementation of Three-Way Split.} 
Copy the elements of $B$ into a temporary set $B'$. Then, for each $a\in \Sigma$:
    \begin{enumerate}
        \item (compute $\delta_a(B)$) Compute $\delta_a(B)$ by scanning the edges $y,x$ labeled $a$ for every $y\in B$, and adding each newly found element to $\delta_a(B)$. To avoid duplicates, mark encountered elements and link them together for later unmarking. During the scan, create and update $\texttt{count}_a(x,B)=|\{y\in B : x\in \delta_a(y)\}|$, and let $x$ point to this record.
        \item (refine $\PPP$ by $B$ w.r.t.\ $\delta_a$) Split each block $D\in\PPP$ that contains some element of $\delta_a(B)$ into $D_1$ and $D_2$. Do so by scanning the elements of $\delta_a(B)$ one by one. For $x\in \delta_a(B)$, determine the block $D$ that contains it, create an associated block $D'$ (if it does not already exist) and move $x$ to $D'$. During the scan, construct and maintain a list of split blocks. After the scan, process the list: mark $D'$ as no longer associated with $D$, delete $D$ if it is empty; if $D$ is not empty and the block of $\mathcal{X}$ containing $D$ and $D'$ has been made compound by the split, add this block to $\mathcal{C}$.
        \item (compute $\delta_a(B)\setminus \delta_a(S\setminus B)$) Scan the edges $y,x$ labeled $a$ where $y\in B'$. To process one such edge, determine $\texttt{count}_a(x,B)$ (to which $x$ points) and $\texttt{count}_a(x,S)$ (to which $y,x$ points). If these integers are equal, add $x$ to $\delta_a(B)\setminus \delta_a(S\setminus B)$ (if not added already).
        \item (refine $\PPP$ by $S\setminus B$ w.r.t.\ $\delta_a$) Proceed as in (b), but by scanning $\delta_a(B)\setminus \delta_a(S\setminus B)$.
        \item (update counts) Scan the edges $y,x$ labeled $a$ where $y\in B'$. Process them one by one, by decreasing $\texttt{count}_a(x,S)$ (to which $y,x$ points) by one. If the count becomes zero, delete the record. Make $y,x$ point to $\texttt{count}_a(x,B)$ (to which $x$ points). 
    \end{enumerate}
    Discard $B'$.

\section{Deferred Material for Section~\ref{sec: DFA co-lex order}}

\subsection{Motivation for computing the smallest-width partial order} \label{sec: app: DFA: motivation}

For an NFA $\AAA = (Q, \delta, s)$ and a co-lex order $\prec$ for $\AAA$, the width $p$ of $\prec$ is defined as the size of the largest antichain, i.e., set of pairwise \emph{incomparable} states, where we say that two states $u,v\in Q$ are incomparable if neither $u\prec v$ nor $v\prec u$ holds. According to Dilworth's theorem~\cite{dilworth1987decomposition}, the width of the co-lex order $\prec$ is equal to the cardinality of the smallest \emph{$\prec$-chain decomposition} of $Q$. Here, a partition $\{Q_i\}_{i\in[k]}$ of $Q$ is a $\prec$-chain decomposition of $Q$, if $Q_i$ is a $\prec$-chain for every $i\in [k]$, where a subset $S\subseteq Q$ is a $\prec$-chain if $\prec$ is a strict total order on $S$. Clearly, the width of a Wheeler order is 1 as no two states are incomparable.

Cotumaccio and Prezza~\cite{CotumaccioP21} call an NFA $\AAA$ \emph{$p$-sortable} if there exists a co-lex order for $\AAA$ that has width $p$. They then show that this parameter $p$ determines the complexity of a series of problems related to storing and indexing the automaton $\AAA$. As an example, they show how to store a $p$-sortable NFA using $\lceil \log |\Sigma| \rceil + \lceil \log p \rceil + 2$ bits per transition by extending the Burrows-Wheeler transform~\cite[Theorem 4.1]{CotumaccioP21} from strings to NFAs. The construction of this data structure however assumes to have the co-lex order $\prec$ and a corresponding $\prec$-chain decomposition of size $p$ as input. Unfortunately, it is NP-hard to determine if a given NFA is $p$-sortable for a given $p$. This follows from the NP-completeness of recognizing Wheeler NFAs, i.e., $1$-sortable NFAs. As a consequence, it is unrealistic to assume that one can compute a co-lex order $\prec$ and a corresponding $\prec$-chain decomposition for a given NFA. For DFAs however, Cotumaccio and Prezza showed that the co-lex order of minimum width is unique and can be computed in $O(|\delta|^2 + |Q|^{5/2})$ time. Kim et al.~\cite{KimOP23} improved this running time by describing two algorithms running in $O(|\delta| \cdot |Q|)$ and $O(|Q|^2 \cdot \log |Q|)$ time on arbitrary DFAs, and an algorithm running in $O(|\delta| \cdot \log |Q|)$ time on acyclic DFAs. All the mentioned algorithms, together with the co-lex order $\prec$ of minimum width $p$, also compute a $\prec$-chain decomposition of size $p$.

\subsection{Infimum and Supremum string}\label{sec: app: DFA: infima suprema}

To define infimum and supremum string, we first have to introduce left-infinite strings that we denote with $\Sigma^\omega$. A \emph{left-infinite string} over the alphabet $\Sigma$ is a sequence of letters from $\Sigma$ of infinite countable length that can be constructed by appending letters to the left of the empty string $\epsilon$. For a string $\alpha\in \Sigma^*$, the notation $\alpha^\omega$ is used to denote the string consisting of an infinite number of copies of the string $\alpha$. 
We denote with $\Sigma^\bullet:=\Sigma^* \cup \Sigma^\omega$ the set of all finite and left-infinite strings.
\begin{definition}[Infimum and Supremum]\label{def: inf sup}
    For a possibly infinite set of finite strings $S\subseteq \Sigma^*$ over an alphabet $\Sigma$, we define 
    \begin{align*}
        \inf S &:= \gamma \in \Sigma^\bullet \text{ s.t.\ } \gamma\ge \beta \text{ for all }\beta \in \Sigma^\bullet \text{ with } \beta \le \alpha \text{ for all }\alpha \in S, \text{ and }\\
        \sup S &:= \gamma \in \Sigma^\bullet \text{ s.t.\ } \gamma\le \beta \text{ for all }\beta \in \Sigma^\bullet \text{ with } \beta \ge \alpha \text{ for all }\alpha \in S.
    \end{align*}
\end{definition}

Notice that the infimum and supremum of a string set $S_v$ for some $v\in Q$ possibly have infinite length, see Figure~\ref{fig: example inf} for a simple example. Intuitively speaking, one can encode $\inf S_v$ (and $\sup S_v$ analogously) using the NFA: The string is obtained by walking backwards from $v\in Q$ by prepending the label of the in-going transition. When there are more than one in-going transitions, we take an \emph{appropriate} branch to yield the smallest string, which can be determined only after sorting the states by their infima strings. When this walk forms a loop then a string with infinite length is obtained.
\begin{figure}[!htb]
\centering
\resizebox{0.3\textwidth}{!}{
\begin{tikzpicture}[->,>=stealth', semithick, auto, scale=.8]
    \node[state, label=above:{}] (S) at (-2,0){$s$};
    \node[state, label=above:{}] (Q_12) at (0,0) {$1$};
    \node[state, label=above:{}] (Q_3) at (2,0)	{$2$};
    \draw (S) edge node[above] {$b$} (Q_12);
    \draw (Q_12) edge node[above] {$a$} (Q_3);
    \draw (Q_3) edge  [loop above] node {$a$} (Q_3);
    
\end{tikzpicture}
}
\caption{An example DFA with a node ($v=2$) whose infimum string $\inf S_v=a^\omega$ is of infinite length.}
\label{fig: example inf}
\end{figure}
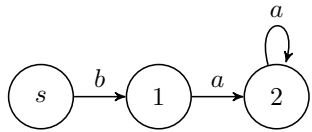

\subsection{Deferred Proofs} \label{sec: app: DFA: proofs}

\LemAAprime*
\begin{proof}
We prove the statement by induction over $i$.
    The base case, i.e., $i=0$, corresponds to the beginning of the first iteration, then $Q$ is the only element in $\XXX$ and thus $Q=A=A'$.
    Let us now assume that the claim holds for $i\ge 0$, i.e., $v\in\delta^{(i)}_a(P)\cap\delta_a(P')$  for some $P,P'\in\XXX^{(i)}$ implies $P = P'$ or $P < P'$.
    We now prove the statement for $i+1$, more precisely, we show that, if $A>A'$ for two parts from $\XXX^{(i+1)}$, then $v\notin\delta^{(i+1)}_a(A)\cap\delta_a(A')$ (which is the contraposition of the claim to be shown). Now let $A$ and $A'$ be two parts from $\XXX^{(i+1)}$ with $A>A'$ and let $S\in\XXX^{(i)}$ be the part chosen in line~\ref{line: choose S} in iteration $i+1$ of the algorithm. There are two cases. (1) Assume that at least one of $A$ and $A'$ is not contained in $S$. Then, there are two parts $P$ and $P'$ in $\XXX^{(i)}$ such that $A\subseteq P$ and $A' \subseteq P'$ and $P>P'$. By the induction hypothesis, it follows that $v\notin\delta^{(i)}_a(P)\cap\delta_a(P')$. This yields $v\notin\delta^{(i+1)}_a(A)\cap\delta_a(A')$, as $\delta^{(i + 1)}_a(P)\subseteq \delta^{(i)}_a(P)$, $\delta^{(i+1)}_a(A)\subseteq \delta^{(i+1)}_a(P)$, as well as $\delta_a(A')\subseteq \delta_a(P')$.
    (2) Now assume that both $A$ and $A'$ are contained in $S$. Then one of $A$ and $A'$ must be $B$ and the other one  $S\setminus B$. Following Observation~\ref{obs: invariant}, all parts of $\PPP^{(i)}$ are forward-stable with respect to $S$ and hence $v\in \delta^{(i)}_a(A)$ implies $v\in \delta^{(i)}_a(A')$. As a consequence of the pruning step, the transitions from $A$ to $v$ are then deleted in iteration $i + 1$ and hence $v\notin\delta^{(i+1)}_a(A)$ and consequently $v\notin\delta^{(i+1)}_a(A)\cap\delta_a(A')$.
\end{proof}

\LemInfimaEncoding*
\begin{proof}
In addition to $\delta^{(i)}$, $\XXX^{(i)}$, and $\PPP^{(i)}$, let $\prec^{(i)}$ be the strict partial order induced by ordered partition $\PPP^{(i)}$; i.e., for $u,v\in Q$, $u\prec^{(i)} v$ if and only if $u\in A$ and $v\in A'$ for two parts $A$ and $A'$ in $\PPP^{(i)}$ such that $A$ precedes $A'$. For convenience, we also define $\prec^{(0)}=\{(u,v)\in Q^2:\lambda(u)<\lambda(v)\}$. First, we claim that  $\prec^{(i)}$ is a co-lex order of $\AAA^{(i)}$.
\begin{claim}\label{lem: colex prunedDFA}
    The strict partial order $\prec^{(i)}$ is a co-lex order for the automaton $\AAA^{(i)}=(Q, \delta^{(i)}, s)$.
\end{claim}
\begin{claimproof}
    We will show by induction over $i$ that $\prec^{(i)}$ is a co-lex order for $\AAA^{(i)}$. The order $\prec^{(0)}$ corresponding to the partition $\PPP^{(0)}=\langle Q_\epsilon, Q_{a_1}, \dots, Q_{a_k}\rangle$ clearly is a co-lex order for $\AAA^{(0)}$ as Axiom~\eqref{first colex axiom} in Definition~\ref{def: colex order} is satisfied and the premise of Axiom~\eqref{second colex axiom} never holds.
    Assume now that $\prec^{(i)}$ is a co-lex order for some intermediate iteration $i$. We will now prove that the order $\prec^{(i+1)}$ corresponding to $\PPP^{(i+1)}$ is a co-lex order. In order to show that $\prec^{(i+1)}$ is a co-lex order, we have to show that Axiom~\eqref{second colex axiom} is maintained. More precisely, for every pair of states $u, v$ with $\lambda(u)=\lambda(v)=a$ for which $u\prec^{(i+1)}v$, we have to show that $x \prec^{(i+1)} y$ for all distinct predecessors $x$ and $y$ of $u$ and $v$, respectively. Hence, let $u,v\in D$ for $D\in\PPP^{(i)}$, as otherwise the statement follows by the induction hypothesis. As in the algorithm, let $D_{1}=D\cap \delta^{(i+1)}_a(B)$, $D_2=D\setminus D_1$, and $D_{11}= D_1\cap \delta^{(i+1)}_a(S\setminus B'))$ and $D_{12}= D_1\setminus D_{11}$, where $S$ is the compound block from $\XXX^{(i)}$ and $B$ be is the splitter chosen in iteration $i+1$. Notice that due to the pruning step,   $D_{11}=\emptyset$. If $u$ and $v$ are contained in the same out of the two sets $D_{12}$ and $D_2$, nothing is to be shown. Hence, assume that $u$ and $v$ are in different sets, w.l.o.g., say $u\in D_{12}$ and $v\in D_2$. Now let $x$ and $y$ be two distinct predecessors of $u$ and $v$ respectively.
    According to Lemma~\ref{lem: AAprime}, all predecessors of $u$ and $v$, and hence also $x$ and $y$, are in a unique part $A$ and $A'$ of $\XXX^{(i+1)}$. As $u\in D_{12}$ and $v\in D_2$, it follows that $x\in B$ and $y\in S\setminus B$. 
    Let us now first assume that $B$ is the first block in $S$ in line~\ref{line: first or last}, i.e., the decided order of $u$ and $v$ is $u\prec^{(i+1)} v$ as $D_{12}$ precedes $D_2$ in $\PPP^{(i+1)}$. As $B$ precedes all blocks in $S\setminus B$ in $\PPP^{(i)}$, we have that $x\prec^{(i)} y$ and thus also $x\prec^{(i+1)} y$. 
    If, instead $B$ is the last block in $S$ in line~\ref{line: first or last}, i.e., the decided order of $u$ and $v$ is $v\prec^{(i+1)} u$ as $D_2$ precedes $D_{12}$ in $\PPP^{(i+1)}$. 
    As $B$ succeeds all blocks in $S\setminus B$ in $\PPP^{(i)}$, we have that $y\prec^{(i)} x$ and thus also $y\prec^{(i+1)} x$. 
    In summary, property~\ref{second colex axiom} of co-lex orders holds in both cases and thus $\prec^{(i+1)}$ is a co-lex order for $\AAA^{(i+1)}$.
\end{claimproof}

    We can now prove Lemma~\ref{lem: infima encoding} using the fact that $\prec^{(i)}$ is a co-lex order of $\AAA^{(i)}$.
    First, we show that $\alpha^*_u$ is a lower bound on all strings in $S_u$, i.e., for every $\alpha\in S_u$, it holds that $\alpha^*_u\le \alpha$.
    Assume for contradiction that there is a string $\alpha\in S_u$ such that $\alpha<\alpha^*_u$.
    Let $1\le i \le |\alpha|$ be the smallest integer such that, for some $v_1,v_2\in Q$, $\alpha[1..i-1]\in S_{v_2}$, $v_1\in\delta_{\alpha[i]}(v_2)$, and $u\in\delta^*_{\alpha'}(v_1)$ where $\alpha'=\alpha[i+1..|\alpha|]$. Intuitively, $v_1$ is the farthest state (from the back) on the path spelling $\alpha$ that can reach $u$ in the pruned automaton $\AAA^*$. Note that such states $v_1$ and $v_2$ must exist as otherwise $\alpha=\alpha^*_u$. Clearly, $v_1\neq s$ as it has an in-going transition. Now by  
    Lemma~\ref{lem: AAprime}, 
    there exists a state $v_3\in Q$ with $v_1\in\delta^*_{\alpha[i]}(v_3)$. Notice that by the minimality of $i$, it holds that $v_1\notin\delta^*_{\alpha[i]}(v_2)$ and thus $v_2\neq v_3$ and, furthermore, by Lemma~\ref{lem: AAprime}, it holds that $v_3 \prec^* v_2$. 
    Let $j$ be the integer such that $\alpha[j..|\alpha|]$ is the longest common suffix of $\alpha$ and $\alpha_u^*$. Note that $j<i$ must hold. Let $v_2'\in Q$ be a state such that $\alpha[1..j-1]\in S_{v_2'}$ and $v_2\in\delta_{\alpha[j..i-1]}(v_2')$
    and let $v_3'\in Q$ be a state such that $v_3\in\delta^*_{\alpha[j..i-1]}(v_3')$. Since $\alpha<\alpha^*_u$ and $\alpha[j..|\alpha|]$ is their longest common suffix, we either have $v_2'=s$ and thus $v_2'\prec^* v_3'$ or $\lambda(v_2')<\lambda(v_3')$. In the second case, it also follows that $v_2'\prec^* v_3'$ by Axiom~\eqref{first colex axiom} in Definition~\ref{def: colex order}. On the other side, we can deduce $v_3'\prec^* v_2'$ from $v_3\prec^* v_2$ using Property \ref{second colex axiom} of Definition~\ref{def: colex order}, a contradiction. See Figure~\ref{fig: proof lemma17a}.

    \begin{figure}[t]
        \centering
        \includegraphics[page=1]{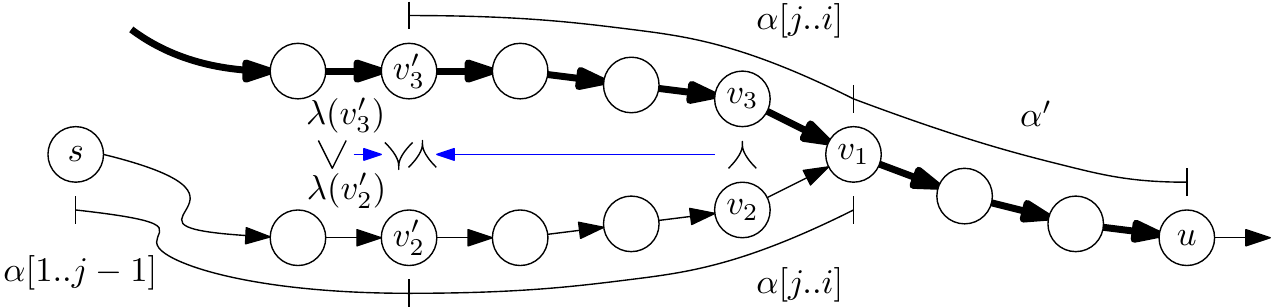}
        \caption{Proof of that $\alpha^*_u$ is a lower bound. Thick transitions indicate $\delta^*$. If we assume that there exists $\alpha\in S_u$ such that $\alpha<\alpha^*_u$, we obtain (from the co-lex property of $\prec^*$) the contradiction that both $v_2'\prec^* v_3'$ and $v_3'\prec^* v_2'$ hold.} \label{fig: proof lemma17a}
    \end{figure}

    Next, we show that $\alpha^*_u$ is the greatest lower bound, i.e., for every lower bound $\alpha$ of $S_u$, it holds that $\alpha\le\alpha^*_u$.
    Assume for contradiction that there exists a lower bound $\alpha$ of $S_u$ such that $\alpha>\alpha^*_u$. 
    If $\alpha^*_u$ is a proper suffix of $\alpha$, then $\alpha^*_u$ must be a finite string from $S_u$. 
    Thus, $\alpha$ is not a lower bound of $S_u$, a contradiction. Now let us consider the case that $\alpha^*_u$ is not a proper suffix of $\alpha$. 
    Let $\alpha_1\in\Sigma^*$ be the longest common suffix of $\alpha$ and $\alpha^*_u$ and let $\alpha_2,\alpha_3\in\Sigma^\bullet$ be the strings such that $\alpha=\alpha_2\cdot\alpha_1$ and $\alpha^*_u=\alpha_3\cdot\alpha_1$. Let $a$ and $b$ be the rightmost character of $\alpha_2$ and $\alpha_3$, respectively. Since $\alpha>\alpha^*_u$, it must hold $a>b$. However there must exist $v\in Q$ such that $u\in\delta_{b\cdot\alpha_1}(v)$ because $\alpha_u^*$ is suffixed by ${b\cdot\alpha_1}$. And since we assume that every state is reachable from the source state $s$, there exists a string $\beta\in\Sigma^*$ such that $v\in\delta_\beta(s)$; consequently, $\beta\cdot b\cdot \alpha_1\in S_u$. However, because $\alpha = \alpha_2\cdot a\cdot\alpha_1 > \beta\cdot b\cdot \alpha_1$, $\alpha$ is no longer a lower bound of $S_u$, a contradiction.
\end{proof}

\subsection{Suffix doubling algorithm on pruned automata} \label{sec: app: DFA: suffix doubling}

Suppose we have computed the pruned automata $\AAA^{\inf}=(Q^{\inf},\delta^{\inf},s^{\inf})$ for infima, and $\AAA^{\sup}=(Q^{\sup},\delta^{\sup},s^{\sup})$ for suprema, according to the algorithm described in Section~\ref{sec: DFA co-lex order}.

Note that, for every $u\in Q^{\inf}$, $\alpha^*_u$ is the infimum string of its associated state on the original input automaton. Similarly, for every $u\in Q^{\sup}$, $\alpha^*_u$ corresponds to the supremum string.
For convenience, we denote by $\hat\AAA=(\hat Q, \hat\delta, \hat s)$ the graph obtained by making two copies of the input DFA $\AAA$, then computing the pruned automaton $\AAA^{\inf}$ on one of them and $\AAA^{\sup}$ on the other one.

At every iteration $k\ge0$, we sort the states $\hat Q$ according to the length-$2^k$ suffixes of their associated strings ($a^*_u$'s).
For the first iteration $k=0$, states are sorted by their incoming label, and each state $u\in \hat Q$ has a pointer $\varphi(u)\in \hat \delta^{-1}_{\lambda(u)}(u)$ to any of its predecessor with respect to the pruned transitions; for the source states, we define $\varphi(s^{\inf})=s^{\inf}$ and $\varphi(s^{\sup})=s^{\sup}$. For next iterations $k\ge1$, we use the relative rank of the length-$2^{k-1}$ suffixes in order to sort the states by the length-$2^{k}$ suffixes of their associated strings. This can be simply done using the fact that the length-$2^k$ suffix of $a^*_u$ is the concatenation of the length-$2^{k-1}$ suffixes of $a^*_{\varphi(u)}$ and $a^*_{u}$. Once we know the relative rank of the length-$2^{k-1}$ suffixes of the states, we can compute the rank for the next iteration using 2-pass radix sort in linear time. Then for each state $u\in \hat Q$, we update its pointer with $\varphi(\varphi(u))$ for further extensions of the suffixes in the next iterations.
Regarding the number of iterations we need to repeat, we can use the following lemma.

\begin{lemma}[\cite{KimOP23}]\label{lem: suffix sort length}
The co-lex order of the infima and suprema strings of a DFA $\AAA=(Q,\delta,s)$ is the same as the co-lex order of their length-$(2\cdot|Q|)$ suffixes.
\end{lemma}

Consequently, after $O(\log |Q|)$ iterations, we obtain the co-lex order of $\{\inf S_u:u\in Q\}\cup\{\sup S_u:u\in Q\}$ associated with the input DFA's states. According to~\cite[Theorem 10]{KimOP23} it characterizes the smallest-width co-lex order $\prec_\AAA$ of the input DFA $\AAA$. Then plugging the algorithm of \cite[Sec.~3.3]{KimOP23}, we can obtain a minimum chain partition of $\prec_\AAA$.

For the time complexity, we can compute the pruned automata each for the infima strings and suprema strings in $O(|\delta|\log|Q|)$ time as shown in Section~\ref{sec: DFA co-lex order}. Then the suffix doubling algorithm described in this section takes $O(|Q|\log|Q|)$ time; because $O(\log|Q|)$ iterations are required, and each iteration requires a 2-pass radix sort and updating pointers, each of which takes $O(|Q|)$ time. Notice that after sorting these strings, we can compute a minimum chain partition of $\prec_\AAA$ in linear time \cite[Section~3.3]{KimOP23} by solving an interval graph colouring problem.

\begin{remark}
One may wonder if our result contradicts the one in \cite[Section 4]{KimOP23}. In particular, (i) a similar algorithm runs in $O(|Q|^2\log|Q|)$ time, and (ii) the merged graph described here is no longer a deterministic and connected automaton, i.e.\ there are two source states, and some states are not reachable from a source state. The difference is that in our case the pruned automata are already computed. In \cite{KimOP23}, each state should keep track of $O(|Q|)$ pointers (instead of only one as in our case) because there are possibly many strings (implicitly) represented by backward walks in intermediate steps that share the same suffixes with the infima and suprema strings. On the other hand, in our case, it is guaranteed that, for each state, the state referred by its pointer is exactly on a walk representing the desired infimum (or supremum) string. As a consequence, we can obtain $O(|Q|\log|Q|)$ time. Regarding the determinism, the only problem caused with the suffix doubling algorithm in \cite{KimOP23} with an NFA is the procedure for updating the pointers. In particular, the union of the updated pointers may take additional time in case of NFAs. As mentioned above, it is also not the case for our algorithm because it is sufficient to keep only one pointer on the (already) pruned automata.
\end{remark}
    
\section{Special Case of String as Input}

Our partition refinement algorithm (\Cref{alg: partition refinement}) can be used to co-lexicographically sort the prefixes of a string $S$ as follows\footnote{We can compute the popular suffix array by co-lexicographically sorting the prefixes of the reverse of $S$.}. We encode $S$ as a path DFA spelling $S$ and run \Cref{alg: partition refinement} in time $O(|\delta|\log{|Q|})=O(|S|\log{|S|})$. The algorithm will output an ordered partition of $Q$, where each part is a singleton set, and thus the order will correspond to the co-lexicographic order of the prefixes of $S$. This approach resembles the algorithm of Baier~\cite{baier:cpm16}, which uses partition refinement on the suffixes of $S$ to obtain the suffix array. However, Baier's algorithm further uses the structural properties of the DFA path (the string $S$) to devise a particular strategy of refinement achieving a linear $O(|S|)$ running time: it starts with the partition induced by the incoming labels (as in \Cref{alg: partition refinement}), but then it proceeds to refine this partition until the point where the suffixes of each part share a well-defined prefixes (called context in the original paper), and finally it uses this property of the current partition to efficiently obtain the final partition (the suffix array). A further difference to our approach is that Baier's method is based on lexicographic order, and thus the choice of the splitter and the split block differs from ours.

\end{document}

%% file: macros.tex
\DeclareMathOperator{\IS}{IS}

\let\epsilon\varepsilon
\let\eps\varepsilon

\newcommand{\PPP}{\mathcal{P}}
\newcommand{\QQQ}{\mathcal{Q}}
\newcommand{\AAA}{\mathcal{A}}
\newcommand{\XXX}{\mathcal{X}}
\newcommand{\BBB}{\mathcal{B}}

\newcommand{\A}{\mathcal{A}}

\usepackage{xspace}

\usepackage[linesnumbered,ruled,vlined,boxed]{algorithm2e}
\newlength{\commentWidth}
\let\oldnl\nl
\newcommand{\nonl}{\renewcommand{\nl}{\let\nl\oldnl}}
\DontPrintSemicolon
\SetKwInOut{Input}{Input}\SetKwInOut{Output}{Output}

\usepackage{tikz}
\usetikzlibrary{decorations.pathreplacing}
\usetikzlibrary{plotmarks}
\usetikzlibrary{positioning,automata,arrows}
\usetikzlibrary{shapes.geometric}
\usetikzlibrary{decorations.markings}
\usetikzlibrary{positioning, shapes, arrows}